\documentclass[10pt,journal,compsoc]{IEEEtran}
\usepackage{filecontents}
\usepackage{amsbsy,amsmath,amssymb,epsfig,bbm,mathrsfs,fancyhdr,fancyvrb,url}
\usepackage{graphicx}
\usepackage{float}
\usepackage{breqn}
\usepackage{chemarr}

\usepackage{amsmath}
\usepackage{amsthm}

\usepackage{enumitem}
\usepackage{mathtools}
\usepackage[utf8]{inputenc}
\usepackage{multirow}
\usepackage{amsfonts}
\usepackage{epsfig}
\usepackage{amssymb}
\usepackage{graphicx}
\usepackage{amssymb,amsmath}
\usepackage{cite}
\usepackage{color,soul}
\usepackage{amsmath}
\usepackage{color}
\usepackage{lipsum}
\usepackage{dblfloatfix}
\usepackage{bbm}
\usepackage{multicol}
\usepackage{subfigure}
\usepackage{xcolor,colortbl}
\usepackage{algcompatible}
\usepackage[noend]{algpseudocode}
\definecolor{Gray}{gray}{0.85}
\usepackage[subfigure]{tocloft}
\usepackage{epstopdf}

\newtheorem{pro}{Proposition}
\newtheorem{re}{Remark}
\allowdisplaybreaks

\usepackage{amsmath, amsthm, amssymb}
\usepackage{bbding}
\usepackage[normalem]{ulem}
\usepackage{tabularx}
\usepackage{multirow}
\usepackage[normalem]{ulem}
\usepackage{array}

\usepackage{dblfloatfix}
\usepackage{multirow}
\usepackage{rotating}
\usepackage{tabularx}
\usepackage{array}
\usepackage{adjustbox,lipsum}

\usepackage{booktabs}
\usepackage{multirow}
\usepackage{tabularx}
\usepackage{array}
\usepackage{dblfloatfix}
\usepackage{hyperref}
\usepackage{blindtext}
\usepackage[normalem]{ulem} 
\usepackage{soul}
\usepackage[linesnumbered,ruled,vlined]{algorithm2e}
\SetKwInput{KwInput}{Input}
\SetKwInput{KwOutput}{Output}

\begin{document}
	
		\title{Cost-Effective Radio and NFV Resource Allocation: E2E QoS Provision
		 }
		\author{\IEEEauthorblockN{Abolfazl Zakeri, \IEEEmembership{Member IEEE}, Narges Gholipoor, Mohammad Reza Javan, \IEEEmembership{Senior Menmber,
				IEEE}, Nader Mokari, \IEEEmembership{Senior Member, IEEE}, and Eduard A. Jorswieck,
			\IEEEmembership{Fellow, IEEE}}
				\IEEEcompsocitemizethanks{\IEEEcompsocthanksitem A. Zakeri, N. Gholipoor, and N. Mokari are with the Department of ECE,
		Tarbiat Modares University, Tehran, Iran (email: \{Abolfazl.zakeri, gholipoor.narges and nader.mokari\}@modares.ac.ir). 
	Mohammad R. Javan is with the Department of Electrical and Robotics Engineering, Shahrood University of Technology,
	Shahrood, Iran (javan@shahroodut.ac.ir).
	Eduard A. Jorswieck is with  TU Braunschweig, Department of Information Theory and Communication Systems, Braunschweig, Germany (jorswieck@ifn.ing.tu-bs.de).
}}
\IEEEdisplaynontitleabstractindextext
\IEEEpeerreviewmaketitle
\IEEEtitleabstractindextext{
	\begin{abstract}
To fend off network ossification and support high degrees of flexibility and various services, network virtualization and slicing are introduced for the next-generation wireless cellular networks.  These two technologies allow diversifying attributes of the future inter-networking and time-varying workloads based resource management paradigms.
In this paper, we propose an end-to-end (E2E) resource allocation framework for future networks considering radio and core network 
by leveraging  network function virtualization (NFV). We aim to minimize the network cost defined based on the utilized energy and spectrum while providing E2E quality of service (QoS) for diverse services with stringent QoS requirements. {This goal is realized by formulating a novel optimization problem which performs the power and spectrum allocation in radio, and service function chaining and
	scheduling in the NFV environment while guaranteeing the distinct QoS constraints of the requested services.
The proposed optimization problem is mixed-integer non-linear programming, which is a non-convex and NP-hard problem. To solve it, we adopt an iterative algorithm with novel admission control and a greedy-based heuristic algorithm, which is shown to have a polynomial order of complexity with $ 13.66 $\% global optimality gap on average for a small scaled network. 
To validate the proposed framework, simulation results are carried out by considering different values of the network parameters and topologies. 
Moreover, our proposed framework and solution algorithm are assessed and compared with the existing works.}
Simulation results demonstrate that the proposed heuristic algorithm and framework outperforms the existing ones by $34$\%  on average in cost reduction. 
\end{abstract}
}
\maketitle
\section{{introduction}}
\IEEEPARstart{T}{o fulfill the proliferation} of the data traffic and various applications requirements, communication service providers (CSPs) need to re-design their infrastructure to support programmability and fine granularity against the rigid networks \cite{8320765}. 	At the same time, CSPs are under pressure to keep up with the capacity
demands and launch differentiated offerings at a short time in a highly competitive service and market. 
Fifth-generation (5G) and beyond are being standardized to meet these requirements by leveraging the network function virtualization (NFV) and softwarization technologies \cite{8125672}.  NFV is introduced as an interesting technology  to reduce the network cost and time to market by virtualizing all the appliances such as servers, routers, storage, and switches \cite{mijumbi2015design, mijumbi2016network,herrera2016resource, 8675284,riera2014virtual}. Moreover, NFV not only provides the commercial off-the-shelf hardware to run a wide spectrum of the virtual network functions (VNFs)\footnote{Examples of VNFs includes firewall, deep packet inspection, transcoding, and load balancing \cite{8125672}. In this paper,  VNF and network function (NF) are the same.} and deploy cloud-native networks and applications
but  it is also a key enabler of  network slicing \cite{alliance2016description} that allows creating multiple logical networks from a physical network \cite{7926921,7243304}. However, some challenges are raised in this area such as   NFV resource allocation and orchestration \cite{herrera2016resource,7243304,riera2014virtual}. 
{This paper focuses on resource management on core and access networks applying NFV and network slicing from network cost perspective.}
 \subsection{Background to NFV and Resource Allocation}
 NFV environment comprises of three entities, namely VNFs, NFV infrastructure (NFVI), and NFV management and orchestration (NFV-MANO) \cite{etsi2014network1, etsi2014network}.
 Note that each  NS consists of  multiple  elements, namely  VNF forwarding graphs (VNF-FGs),  virtual links, physical network functions (PNFs),
 VNFs, and NFVI where they requires a  new and different sets of management and orchestration  functions. 
 Generally speaking, these functions refer to the NFV resource allocation (RA) and orchestration that are  widely appeared in the literature in recent years \cite{7962178,mijumbi2015design,zeng2018stochastic,6782394,herrera2016resource,7938391,8631169}.
\\\indent
 NFV-RA consists of three phases: 1) VNF-FG in which the chaining and the connectivity of the VNFs in an NS is determined, which is also known as the service function chaining (SFC) \cite{7945848}, 2) VNF embedding (placement) in which VNFs are mapped to servers/virtual machines (VMs) \cite{cohen2015near}, and 3) VNF scheduling in which the running time for a VNF under given constraints is determined \cite{7243304}. 
  Each of the above mentioned phases has a pivotal impact on the network performance, 
 its reliability, operation cost, and the experienced QoS. Nevertheless, optimizing all of these phases gives significant reduction in the cost for CSPs and provides various range of services/applications for the end-users and verticals in the shortest time to market.
 \\\indent
 In addition to the NFV resource orchestration in the cloud-based core network, the access network plays a key role in the QoS provisioning and user experience as well as the network OpEx, and has some impacts on the NFV resource allocation regarding the cloud-radio access network (cloud-RAN) and {the} generated traffic \cite{hossain20155g, 7143328, 8004168,7143328}. {Note that 
 proposing an NFV and radio RA framework to provide E2E QoS for end-users and network cost reduction is the main focus of this paper.}
\subsection{{\textcolor{black}{ Related Works}}}
{ 
	We provide a review on related works which are
	categorized into three groups, namely NFV, network slicing, and network cost model.} 
\subsubsection{NFV-RA}
	{Based on the previous discussions on NFV, we can further divide  NFV-RA into scheduling, embedding, and SFC problems as discussed in the following. }
	\paragraph{VNF-Scheduling} 
In \cite{7490404}, a VNF scheduling problem is investigated and a joint VNF scheduling and traffic steering problem is formulated as a mixed-integer nonlinear problem (MINLP). A low computational complexity matching-based algorithm is devised for online VNF scheduling in \cite{game-theroy}. The authors in \cite{7502870} study the VNF scheduling by formulating a 
	MILP problem whose objective is to minimize the latency of all VNFs. 
	They adopt a genetic algorithm to solve the optimization problem in a low complexity manner. In our proposed scheduling model, the processing latency is captured  from the resources given to each VNF and the amount of bit rate passing through the VNF compared to the fixed processing latency that is assumed  in \cite{riera2014virtual, mijumbi2015design, 7490404, game-theroy, 7502870}.
	\paragraph{VNF-Embedding/Mapping/Placement}
	In \cite{cohen2015near}, the problem of NF placement is studied
		and the cost of having VMs\footnote{In this paper, VM, node, and server have the same meaning, and the cost of steering the traffic into the servers are investigated.} {are studied.}
		  An automated decentralized method for online placement and optimization of VMs in NFV-based network is proposed in \cite {8501940}.
	In \cite{8281644},  VNF embedding with the aim of minimizing  time-varying workloads of physical machines is studied. Furthermore, users' SFC requests and factors such as basic resource consumption and time-varying workload are {taken} into consideration.
	 The authors in \cite{7859379} formulate a joint operational and traffic cost optimization problem whose goal is finding a cost-efficient VNF placement algorithm. To solve it, they propose a modified version of  
	the Markov approximation technique that is a combination of Markov and matching algorithm. The reason behind the proposed approach is that the Markov approximation suffers from a long-time convergence, and cannot be applied in practice for large networks.
	\paragraph{VNF-SFC}
	In an NFV-based network, each service consists of a set of NFs that need to be executed in a specific order to provide the service, which is called SFC. Each SFC is the heart of the insertion of particular business services into the network and  its simple definition is linkage of NFs to form a NS. The SFC problem in NFV context is widely investigated in the literature \cite{7417401,liu2017dynamic,8937740,7945848}.
	A dynamic SFC deployment is proposed in \cite{liu2017dynamic} in which the authors consider a trade-off the between resource consumption and operational overhead. The authors in \cite{8937740}  study the reliability concerns with reducing the experienced delays by incorporating the VNF-decomposition-based backup strategies into a MILP  problem. 
	\subsubsection{Combination of NFV-RA Phases and Network Slicing}
In \cite{mijumbi2015design},  an online scheduling and embedding algorithm is proposed in which the capacity of the available buffers and the processing time of each VNF is considered. 
 The authors propose a set of greedy-based algorithms for mapping and scheduling. Moreover, the cost, revenue, and acceptance ratio of these algorithms are compared.
 {The} VNF placement in a network with several mobile virtual network operators (MVNOs)  is investigated in  \cite{riggio2016scheduling}  in which a slice scheduling mechanism is introduced to isolate the traffic flow of MVNOs and optimize {the} VNF placement based on the available radio resources. {The} joint VNF placement and admission control (AC) with maximizing the network provider revenue in terms of bandwidth and capacity {are} studied in \cite{nejad2018vspace}.  
 The authors in \cite{kim2018performance} propose an RA algorithm that integrates the placement and scheduling of VNFs.  
{In \cite{8460139}, a framework for providing on-demand network slicing with leveraging softwarization and virtualization technologies is proposed.  The authors of that work define each slice as a SFC with a specific life cycle and end-to-end (E2E) delay as the key performance indicators. Then the allocation of the requested resources of each slice is formulated as an optimization problem in which the goal is to minimize the cost of the resource utilization. }

The authors in \cite{8845306} propose a new monitoring architecture that is a local entity called monitoring agent with an eye on the hierarchical architecture to orchestrate the network resources taking into account the resource demands in terms of slices. Moreover, they propose a new protocol for monitoring service status of local agents. However, they do not pay attention to the E2E resource orchestration, and mainly focus on the scalability of the architecture to reach a flexible network with negligible network overhead.
{The} authors in \cite{8647504} formulate a MILP optimization problem to orchestrate the underlying resources according to the users' requirements in a cost-efficient manner. Since the considered framework is based on the network slicing, they assume that each slice has a set of SFCs with some NFs that are virtually interconnected.  Actually, they propose a framework that maps the virtual SFCs into the underlying/physical networks under resources, links, and latency  constraints. However, the main drawback of their work is that they do not consider the processing and waiting time in the latency.
Therefore, {this approach is only} appropriate for core network slicing.  
\subsubsection{Cost Model}
Cost-saving (OpEx and CapEx) solutions are important for future networks not only for vendors and infrastructure providers but also for CSPs from revenue maximization aspect. Therefore, many works have appeared in this area \cite{8647504, kim2018performance, 7835175, cohen2015near, 8460139, 7934437, 7859379, 8247219}. One branch of studies,  introduces energy efficiency in cloud-enabled data centers \cite{8647504, 7934437}, and IoT networks \cite{8638582}.  
	In particular,  \cite{8672634} {formulates} the total cost of links and nodes by defining the link connection cost and VNF setup cost, and then investigates an embedding and routing policy such that the network cost is minimized. The objective of  \cite{cohen2015near} is to minimize the total system cost under the allocation of functions to nodes and assignment of clients to functions. The cost is defined based on the allocated resources of machines to functions. The authors of \cite{8647504} {aim}  to minimize the number of nodes hosting the NFs under the placement, latency, and bandwidth constraints. 
	Moreover, \cite{7949048} focuses on the energy saving in the cloud-RAN by determining the radio unit sleep scheduling and VM consolidation strategies. To this end, the authors define a general form of the power consumption which includes the active and static consumed power. To save more energy, they also {assume that if there is no} user to serve in radio units, VMs of the radio units are shut down.
\\\indent
  As previously mentioned,  one objective  of this paper is  to diminish {the} network cost defined as {the} number of active/on VMs in cloud-nodes and the radio power and spectrum. 
	The motivation behind this is that we study the radio and NFV cost alongside with each other to investigate the effect of each of which on different service types, total network cost, and {modeling of} the network slicing cost (see simulation results, especially, Fig. \ref{NC}).
	 For example,  high data rate services have more effect on the radio cost than  low latency services.   The works in  \cite{7949048} and \cite{8647504} have a similar approach to this paper in energy efficiency modeling as they {aim} to minimize {the} total power consumption in a cloud-RAN and {the} number of nodes hosting NFs.
	  However, their considered approach is not a practical model for taking into {account the}  total network cost, since the total network cost not only includes OpEx (can be considered {the utilized} energy in the whole of the network, i.e., radio and core) but also CaPex cost, especially, the spectrum acquisition cost. Therefore, these approaches have drawbacks for applying {the} total network cost or slice provisioning cost. 

{In order to summarize,} none of them investigate  NFV-RA jointly with the radio RA in an E2E  QoS-aware framework. However, in \cite{9000731}, the authors study a wireless-based NFV in the context of cross-layer resource allocation. This most recent letter optimizes the number of resource blocks that are assigned to SFCs, links, and nodes. 
The objective of the study is to minimize the total delay. As can be inferred, our work has main differences compared to \cite{9000731} in terms of power optimization, scheduling, and considering cost with E2E QoS requirements.
At the same time, the main drawback of the previous works is that they do not consider the cost of the radio and NFVI resources in the total network cost and performance.
Obviously, in a real network, {for} deploying the network slicing, providing E2E QoS is pivotal in user-experience and network performance and adds more challenges than NFV resource orchestration does. The above discussion motivates us to introduce a novel E2E RA including NFV and radio in a cloud-enabled network aiming to minimize total network cost. 
\subsection{\textbf{Contribution and Research Outcome}}
  	 The main contributions of this paper can be summarized as follows:
\begin{itemize}
	\item {We propose a novel E2E QoS-aware framework by taking into account both the radio and NFV-RA  in an unified E2E RA optimization problem.  
	The ultimate goals of the proposed E2E framework are realizing the E2E network slicing and providing cost-efficient E2E QoS for the customers.}
	\item 	We formulate a new optimization problem for radio and NFV-RA called \textbf{j}oint \textbf{r}adio and \textbf{N}FV-\textbf{RA} (JRN-RA) with the aim of minimizing the  E2E network cost in terms of the utilized  radio resources, i.e., power and spectrum, as well as the number of active VMs/servers.
		By this approach,  we can reduce the consumption of power in the active state and at the same time save power in the idle mode (e.g., low power state).
	Hence, it would be a more effective way to achieve more energy saving.
	\item To overcome the restriction on the network resources which may cause the	 infeasibility of the optimization problem,	we devise a novel elastication-based AC algorithm and iterative/multi-stage approach to  solve the proposed optimization problem.
		\item By applying  the iterative approach, we solve iteratively the radio and NFV-RA sub-problems. The NFV-RA  sub-problem is a non-linear integer programming problem. 
			To solve it, we propose a greedy-based low complexity algorithm whose aim is minimizing the number of active VMs/servers. Moreover, we compare its performance with the state of the art works, e.g., \cite{mijumbi2015design}. 
			\item 
		 We assess and compare the performance of the   proposed framework with 
		 \cite{7949048, 8647504} in Section \ref{Bench_Mark} from the defined cost and acceptance ratio point of views. In numerical results, we show that the proposed system model outperforms the state of the art. 
	\item We prove the convergence of the solution of the JRN-RA problem and analyze its computational complexity. Moreover, we investigate the optimality gap of the proposed iterative algorithm. Numerical results reveal that the proposed solution has a polynomial order of complexity with an acceptable optimality gap as $ 13.66$\% on average.
\end{itemize}
\subsection{\textbf{Paper Organization}}
The rest of the paper is outlined as follows.  In Section \ref{systemmodel},  the system model and problem formulation {are} explained. 
The proposed solution is presented in Section \ref{solutions}.  In Section \ref{Complexity-Convergance}, the computational complexity and convergence of our solution are discussed. The simulation results are presented in Section \ref{simulations}. Finally, in Section \ref{conclusions}, the conclusion remarks is inferred.

\textcolor{black}{\textbf{Symbol Notations:} Vector and matrices are indicated by bold
lower-case and upper-case characters, respectively. 
 $\mathcal{A}$ denotes set $\{1,\dots,A\}$, $\mathcal{A}(i)$ is {the} $i$-th element of set $\mathcal{A}$,  and $\Bbb{R}^{n}$ is the set of $n$ dimension real numbers.  Moreover, $U_d[a,~b]$ denotes the uniform distribution in interval $a$ and $b$ and $ |.| $ indicates absolute value.} 
\section{{system Model and problem formulation}} \label{systemmodel}
\textcolor{black}{We consider E2E network of an operator in which it comprises radio access and core network with an access point and some NFV-enabled cloud nodes/servers as shown in Fig. \ref{pic}.  Details of the considered system model and its parameters are stated in two main parts, i.e., access and core network descriptions as follows.
	{It can be  noticed that wireless channel states in the radio domain change rapidly,  while parameters of NFV (e.g., SFC) change relatively slow. In this paper, since  we do not study the long-term optimization, we assume that the network variant parameters in the radio (e.g., channel information) and NFV (e.g., SFC) parts are fixed in our optimization problem as in the existing works \cite{9000731, 7949048}. This means that we solve the optimization problem with the given parameters where some of these parameters. e.g., channel gains and the requested service's SFC, are changed in the adopted Monte Carlo simulation method.}
		}
\begin{figure*}
	\centering
	\centerline{\includegraphics[width=0.65\textwidth]{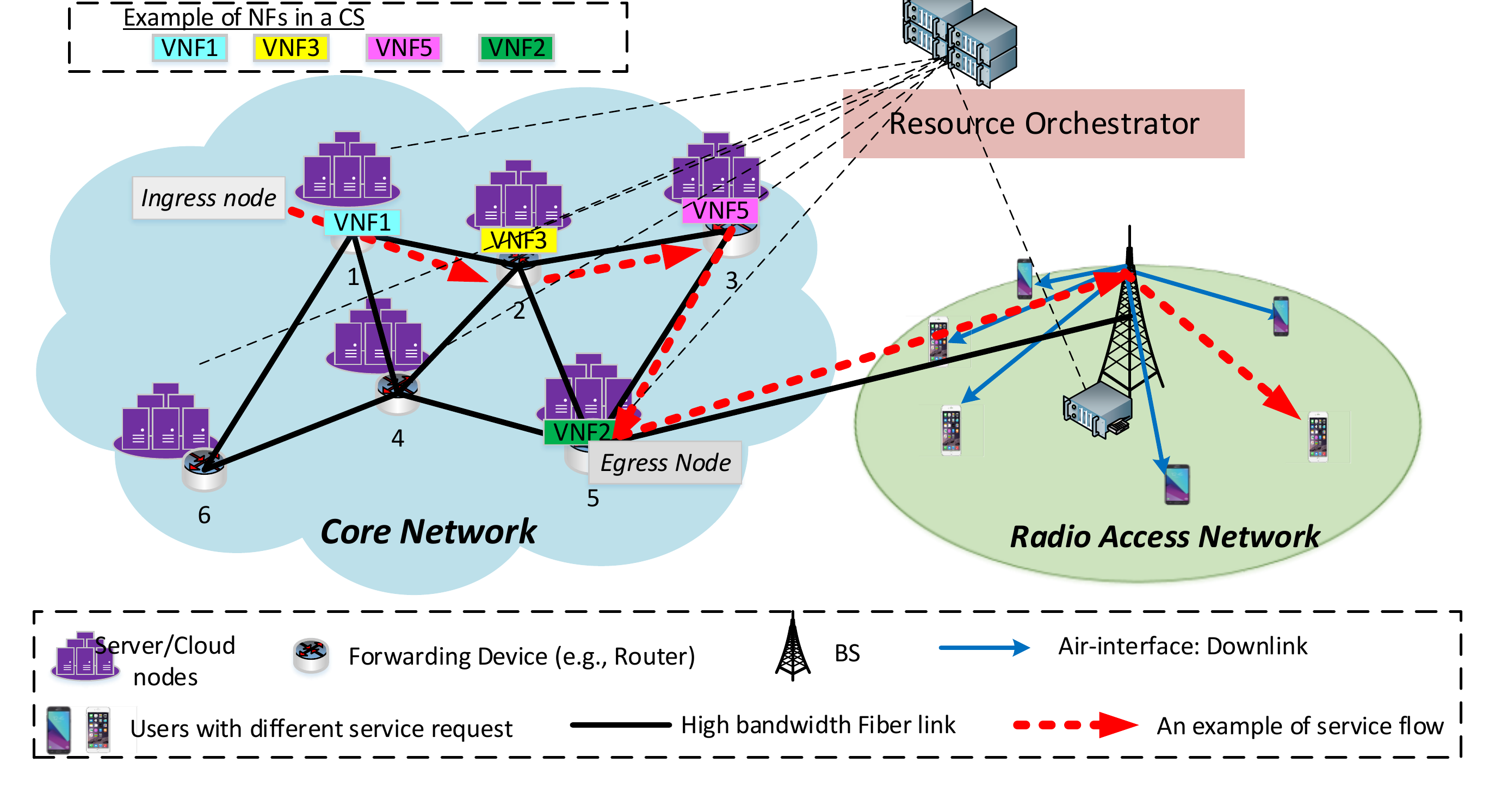}}
	\caption{High level presentation of considered E2E network of a operator  with an example E2E flow  for a user is denoted by red dashed line. All of optimization variables are derived by the entity called resource orchestrator. 
  }
	\label{pic}
\end{figure*}
\subsection{\textbf{Radio Access Network Description}}
We consider a single-cell with  a set $\mathcal{U}=\{1,\dots,U\}$ of $U$ users
and a set     $\mathcal{K}=\{1,\dots,K\}$ of $K$ subcarriers with subcarrier spacing $ B $. 
 We define the subcarrier assignment variable $\rho_{u}^{k}$ with $\rho_{u}^{k}=1$ if subcarrier $k$ is allocated to user $u$ and otherwise $\rho_{u}^{k}=0$. We assume orthogonal frequency division multiple access (OFDMA) as the transmission technology in which each subcarrier is assigned at most to one user. To consider this, the following constraint is \textcolor{black}{introduced}:
\begin{equation} \label{subcons}
\sum_{u \in \mathcal{U}} \rho_{u}^{k} \le 1, \forall k\in\mathcal{K}.
\end{equation}
Let $h_{u}^{k}$ be the channel coefficient between user $u$ and the
BS on subcarrier $k$, $~p_{u}^{k}$ be the transmit power from the BS to user $u$ on subcarrier $k$, and $\sigma_{u}^{k}$ be the power of additive white Gaussian noise (AWGN)\textcolor{black}{\footnote{\textcolor{black}{In this paper, we assume that an AWGN interfering source (IS) interferes at the BS and all users on each subcarrier. We consider a single cell with a BS, in a scenario with many cells and no coordination between BSs, the inter-cell interference distribution converges to a Gaussian and can be integrated into the interference of other cells which can be modeled by the IS \cite{7833146}.}}}  at user $u$ on subcarrier $k$. 
The received signal to noise ratio (SNR) of user $u$ on
subcarrier $k$ is $ \gamma_{u}^{k}=\frac{ p_{u}^{k} h_{u}^{k}}{\sigma_{u}^{k}}$, and the achievable data rate (in bits per second/Hz) of user $u$ on subcarrier $k$ is given by
\begin{align}
r_{u}^{k}= \rho_{u}^{k}\log(1+\gamma_{u}^{k}), \forall u \in \mathcal{U}, k \in \mathcal{K}.\label{Rate_for}
\end{align}
Hence, the total achievable rate of user $u$ is given by
$R_{u}=\sum_{k \in \mathcal{K}}r_{u}^{k},\,\forall u\in\mathcal{U}$.
The power limitation of BS is
$\sum_{k\in \mathcal{K}}  \sum_{u\in \mathcal{U}} \rho_{u}^{k}  p_{u}^{k}  \le P_{\max},$ 
 where $P_{\max}$ is the maximum transmit power of BS.
\subsection{{NFV Environment Description}}
In this subsection, we explain how the generated traffic of  each user is handled in the network by performing different NFs in the requested user's NS\footnote{Defined by European Telecommunications Standards Institute (ETSI) as the composition of Network Function(s) and/or Network Service(s), defined by its functional and
	behavioral specification \cite{ETSIG003}.} on the different servers/physical nodes by leveraging NFV\footnote{Standardized by  ETSI organization for 5G and beyond\cite{etsi2013network}.}. In this regard, we consider NFV-RA that consists of a new approach for the embedding and scheduling phases. In the embedding phase, \textcolor{black}{we map each NF on the server that is capable to run that NF.} \textcolor{black}{Note that we do not consider mapping virtual links on the physical links 
	and leave it as an interesting future work as \cite{mijumbi2015design, yoon2016nfv}.}

We consider $S$ communication service (CS)\footnote{\textcolor{black}{In this paper, the NS and CS are paired together. That means each CS $s$ has a NS with  corresponding NFs that is denoted by set $\Omega_{s}$. Note that CS is defined by the 3rd generation partnership project (3GPP) technical specification 28.530 \cite{28530}.}}
types whose set is $\mathcal{S}=\{1,2, ..., S\}$ and $M$ NFs whose set is  $\mathcal{F} = \{f_{m}~\big|~m=1,\dots,M\}$. The considered parameters of the paper are stated in Table \ref{Table_parameters}.
\begin{table}[t]	\label{Table_parameters}
	\renewcommand{\arraystretch}{1.5}
	\centering
	\caption{Network parameters and notations}
	\begin{adjustbox}{width=.48\textwidth,center}	
	\begin{tabular}{| c| l| }	
		\hline
		\textbf{Notation}& \textbf{Definition}\\\hline
		$\mathcal{U}/U/u$&Set/number/index of users\\ \hline
		$\mathcal{N}/N/n$ &Set/number/index of  VMs\\ \hline
		$\mathcal{F}/F/f$ &Set/number/index of  NFs\\ \hline
		$\mathcal{S}/S/s$ &Set/number/index of  NSs\\ \hline
		$P_{\max}$&Maximum transmit power of the BS	\\\hline
			${\rho_{u}^k}$ &Assignment of subcarrier $k$ to user $u$ \\\hline
			${p_{u}^k}$ &Transmit power of user $u$ on subcarrier $k$ \\\hline
		${h_{u}^k}$ &Channel coefficient between user $u$ and the BS on subcarrier $k$\\\hline
		$\gamma_{u}^k$&SNR of user $u$ on subcarrier $k$\\\hline
		$r_{u}^k$ &Achieved rate of user $u$ on subcarrier $k$\\\hline	
		$y_{u}$ & Packet size of the requested service of user $u$ \\\hline
		$\alpha^{f_{m}^{s}}$, $\psi^{f_{m}^{s}}$& Processing and buffering demand of NF $f_{m}$ in NS $s$, respectively\\\hline
		$\tilde{\tau}_{n}^{f_{m}^{s}}$& Processing latency of NF $f_{m}$ on server $n$ in NS $s$\\\hline
		$L_n, \Upsilon_n$ & Processing and buffering capacity  of VM $n$, respectively\\\hline
		$\eta_{n}$ & Server $n$ activation indicator \\\hline
		$\beta_{u,n}^{f_{m}^{s}}$ & Server mapping between  NF $f_{m}^{s}$ of service $s$ for user $u$, and node $n$\\\hline
		$t_{u,n}^{f_{m}^{s}}$ & Starting time of NF $f_{m}^{s}$ of service $s$ which is requested by user\\& $u$ at node $n$\\\hline
		$x_{u,u'}^{f_{m}^{s},f_{m'}^{s'}}$& Ordering indicator between NF $f_{m}^{s}$  of service $s$ for user $u$ and\\ & NF  $f_{m'}^{s'}$ of service $s'$ for user $u'$\\\hline
	\end{tabular}
\end{adjustbox}
\end{table}	
{
Each CS $s$ is defined by the tuple $\text{CS}_{s}={\Big(}\textit{source node}, \textit{destination node},\Omega_{s},D^{\text{max}}_{s},R^{\min}_{s}{\Big)}$ where $\Omega_{s}$ is the set of NFs which constructs NS $s$ defined by $\Omega_{s}=\left\{f_{m}^{s}\right\},\,m\in\{1,\dots,M\}$, $D^{\text{max}}_{s}$ is the latency constraint for each packet of NS $s$, and $R^{\min}_{s}$ is the minimum required data rate of NS $s$.
We assume that each user can request at most one CS at a time. It can be readily noticed that each CS actually is a slice which includes
\begin{align}
\label{Slice_Re}
\Big(\underbrace{n_s^i, n_s^o,\Omega_{s}}_{\text{SFC of each slice}},\underbrace{D^{\text{max}}_{s},R^{\min}_{s}}_{\text{QoS requirements}}{\Big)}.
\end{align}
where $ n_s^i $ and $ n_s^o $ denote the source and destination nodes of service $ s $.
 Observe that \eqref{Slice_Re} is practical and really forms a slice
which is requested by an end-user compared with \cite{8460139} which considers SFC as a slice. Notably, some users may request the same slice.}

We consider a set of VMs denoted by  $\mathcal{N} = \{1, ..., N\}$ in the network each of which has a limited amount of computing and storage resources.
We assume that each   server can process at most one function at a time \cite{mijumbi2015design}, but it can process any NF \cite{mijumbi2015design}, if capable to run it. This processing approach occurs sequentially for NFs as the time elapses. We consider a generalized model for processor sharing of VMs that is introduced in  \cite{mijumbi2015design}.

To design an energy efficient framework for NFV-environment  in our proposed system, we introduce a new variable $\eta_{n}$ to determine the active nodes, which is defined as 
\begin{equation} \nonumber
\begin{split}
&\eta_{n}= \begin{cases}
1, & \text{Node $n$ is active},\\
0, & \text{Otherwise}.
\end{cases}
\end{split} 
\end{equation}
The goal of our work is to minimize the total number of active VMs/servers in the network. The gain of this approach is not only saving the consumption of the power in the active mode (i.e., under load) but also saving the power of the servers in the idle mode. 

We introduce a binary variable $\beta_{u,n}^{f_{m}^{s}}$ (i.e., VNF-placement variable) which denotes that  NF $f_{m}^{s}$ for user $u$ in NS $s$ is executed at node $n$, and is defined as
\begin{equation} \nonumber
\begin{split}
&\beta_{u,n}^{f_{m}^{s}}= \begin{cases}
1, & \text{NF~ $f_{m}^{s}$~for~$u$~in NS~$s$~ is executed at server~$n$}.\\
0, & \text{Otherwise}.
\end{cases}
\end{split} 
\end{equation}
When $\beta_{u,n}^{f_{m}^{s}}$ is $1$, 
 \textcolor{black}{server $n$} should be active, i.e.,  $\eta_n =1$. Therefore, we have the following constraint:
\begin{equation} \label{sercons}
\beta_{u,n}^{f_{m}^{s}} \le \eta_{n},\, \forall n\in\mathcal{N},\forall u\in\mathcal{U}, \forall f_{m}^{s}\in\Omega_{s}, \forall s\in\mathcal{S}.
\end{equation}
Each NF of each NS is performed completely at only one VM at a time \cite{7490404}. Therefore, we have
\begin{equation} \label{maxfscon}
\sum_{n \in \mathcal{N}}\beta_{u,n}^{f_{m}^{s}} \le 1, \forall u \in \mathcal{U}, f_{m}^{s} \in \Omega_{s}, s\in\mathcal{S}. 
\end{equation}	
\\\indent
	Moreover, we assume  that each NF needs a specific number of CPU cycles per bit, i.e., $\alpha^{f_{m}^{s}}$,  to run on the assigned server. From the physical resource perspective, we assume that each server $n$ can provide at most $L_n$ CPU cycles per unit time, and hence, we have the following constraint:
\begin{equation} 
\label{pcapcons}
\sum_{u \in \mathcal{U}}
\sum_{s\in\mathcal{S}}\sum_{f_{m}^{s}\in\Omega_{s}}
{y}_{u}\alpha^{f_{m}^{s}} \beta_{u,n}^{f_{m}^{s}} \le L_n, \forall n \in \mathcal{N}, 
\end{equation}
where $y_{u}$ is the packet size of the service of user $u$. Here, we assume that the packed size is equal to the number of bits generated in a unit time.
Hence, the elapsed time of each NF $f_{m}^{s}$ for each bit on server $n \in \mathcal{N} $ 
	is obtained as follows: 
\begin{equation} \label{deleq}
\tilde{\tau}_{n}^{f_{m}^{s}}= \frac{\alpha^{f_{m}^{s}}}{L_n},\forall n\in\mathcal{N}, f_{m}^{s} \in \Omega_{s}.
\end{equation} 	 
Therefore, the total processing latency of running  NF $f_{m}^{s}$ on server $ n $ for each packet with packet size 
${y}_{u}$
is obtained as 
\begin{align}\label{Final_Del}
\tau_{n}^{f_{m}^{s}}=	\tilde{\tau}_{n}^{f_{m}^{s}}{y}_{u},\forall n\in\mathcal{N}, f_{m}^{s} \in \Omega_{s}.
\end{align}
\textcolor{black}{Additionally, we assume that each NF needs specific storage size, i.e., $\psi^{f_{m}^{s}}$, when it is running on the server. Moreover, each packet  consumes   $y_{u}$ buffer capacity, when it is waiting for running a NF on an assigned server.  Hence, from the storage and buffer resource perspective, we consider that each server has a limited buffer and storage size, i.e., $\Upsilon_n$, which leads to the following constraint:}
\begin{align}
	&\sum_{u \in \mathcal{U}}\sum_{s\in\mathcal{S}}\sum_{ f_{m}^{s}\in\Omega_{s}} (\psi^{f_{m}^{s}}  +
	{y}_{u} )\beta_{u,n}^{f_{m}^{s}}\le \Upsilon_n ,\label{bufcons}\forall n \in \mathcal{N}.
\end{align}
\subsection{\textbf{Latency Model}}
In NFV-RA, our main aim is to guarantee the service requirement\textcolor{black}{, which} includes maximum tolerable latency  for each packet with size ${y}_{u}$ of the requested services \textcolor{black}{while} minimizing the energy consumption of VMs.  
The total latency that we consider in our system model results from executing NFs and queuing (waiting) time. In the following, we calculate the total latency resulting from scheduling.
\begin{re}
	In this paper, our main aim is to model and investigate the effect of processing and scheduling latency on the service acceptance and the network cost. Hence, we do not consider the other latency factors such as propagation and transmission latency in our model. In fact, our proposed scenario  is focused on intra data center communications  and not appropriate for the national-wide networks.
	It is worth noting that the aforementioned latency is coming from the high order distance from the source and application servers. Therefore, these concernes can be treated by exploiting the mobile edge computing (MEC) and content delivery networks (CDNs) technologies to bring the application servers close to clients \cite{fajardo2015improving}. \textcolor{black}{The extension of this work to MEC-enabled networks is beyond the scope of the current paper, but planned in future works.}
\end{re}
\subsubsection{\textbf{Scheduling and Chaining}}
\textcolor{black}{Each NF should wait until its preceding function is processed before its processing can commence.
The processing of NS $s$ ends when its last function is processed.} Therefore, the total processing time is \textcolor{black}{ the summation of the processing times of the NFs at the various servers. 
	For scheduling of each NF on a server,}  we need to determine the start time of it. Therefore, we define $t_{u,n}^{f_{m}^{s}}$ which is the start time of running NF $f_{m}^s$ of the requested service $s$ for user $u$ on server $n$. Furthermore, we introduce a new variable $x_{u,u'}^{f_{m}^{s},f_{m'}^{s'}},$ 
in which, if NF $f_{m}^{s}$ of user $u$ is running after NF $f_{m'}^{s'}$ of user $u'$, its value is $1$, otherwise is $0$. 
\textcolor{black}{By these definitions, the starting time of each NF can be obtained as follows:
	\begin{align}
	\label{schule}
	& t_{u,n}^{f_{m}^{s}}\beta_{u,n}^{f_{m}^{s}}\ge  \max \Bigg\{ \mathop {\max }\limits_{\scriptstyle\forall f_{m'}^{s'}\in\Omega_{s'},
		\scriptstyle u'\in\mathcal{U}} \left\{ {x_{u,u'}^{f_m^{s},f_{m'}^{s'}}\beta _{u',n}^{f_{m'}^{s'}}(t_{u',n}^{f_{m'}^{s'}} + \tau _n^{f_{m'}^{s'}})} \right\},
	\nonumber\\ &
	\mathop {\max }\limits_{\scriptstyle\forall f_{m''}^s\in \{{\Omega _s}-f_{m}^{s}\},
		\scriptstyle  n'\in \{\mathcal{N}-n\} } \left\{ {x_{u,u}^{f_{m}^{s},f_{m''}^{s}}\beta _{u,n'}^{f_{m''}^s}(t_{u,n'}^{f_{m''}^{s}} + \tau _{n'}^{f_{m''}^{s}})} \right\}\Bigg\},
	\nonumber\\&
	\forall f_m^s \in {\Omega _s}, f_{m'}^{s'}\in{\Omega_{s'}},{\mkern 1mu}  \forall s, s'\in\mathcal{S}, \forall n \in \mathcal{N},\, \forall u\in \mathcal{U}. 
	\end{align}}
 
	{
		To demonstrate  how we formulate the scheduling of NFs, the proposed scheduling policy is illustrated  in Fig. \ref{Scheduling}.  
This figure is the state of the network  assuming two NSs each of which consists of some NFs  and five VMs.
} The processing time of each NF, $\tau_{n}^{f^{s}_{m}}$ is obtained by \eqref{Final_Del}, i.e., $\tau_{5}^{f^{2}_{1}}=\frac{\tilde{\tau}^{f^{2}_{1}}y_{u}}{L_{5}}$. Each NF has a start time (denoted by $t_{u,n}^{f_{m}^{s}}$) to run on an assigned VM/server and elapsed processing time (denoted by $\tau_{n}^{f^{s}_{m}}$) and 
is completed by the time  given by $t_{u,n}^{f_{m}^{s}}+\tau_{n}^{f^{s}_{m}}$ on server $n$. As can be seen, 
VMs $3$ and $1$ are off, since based on our aim and solution algorithm, three VMs from five VMs are sufficient to ensure the requested requirements.	
%
%
%
	\begin{figure*}[t]
		\centering
		\centerline{\includegraphics[width=0.75\textwidth]{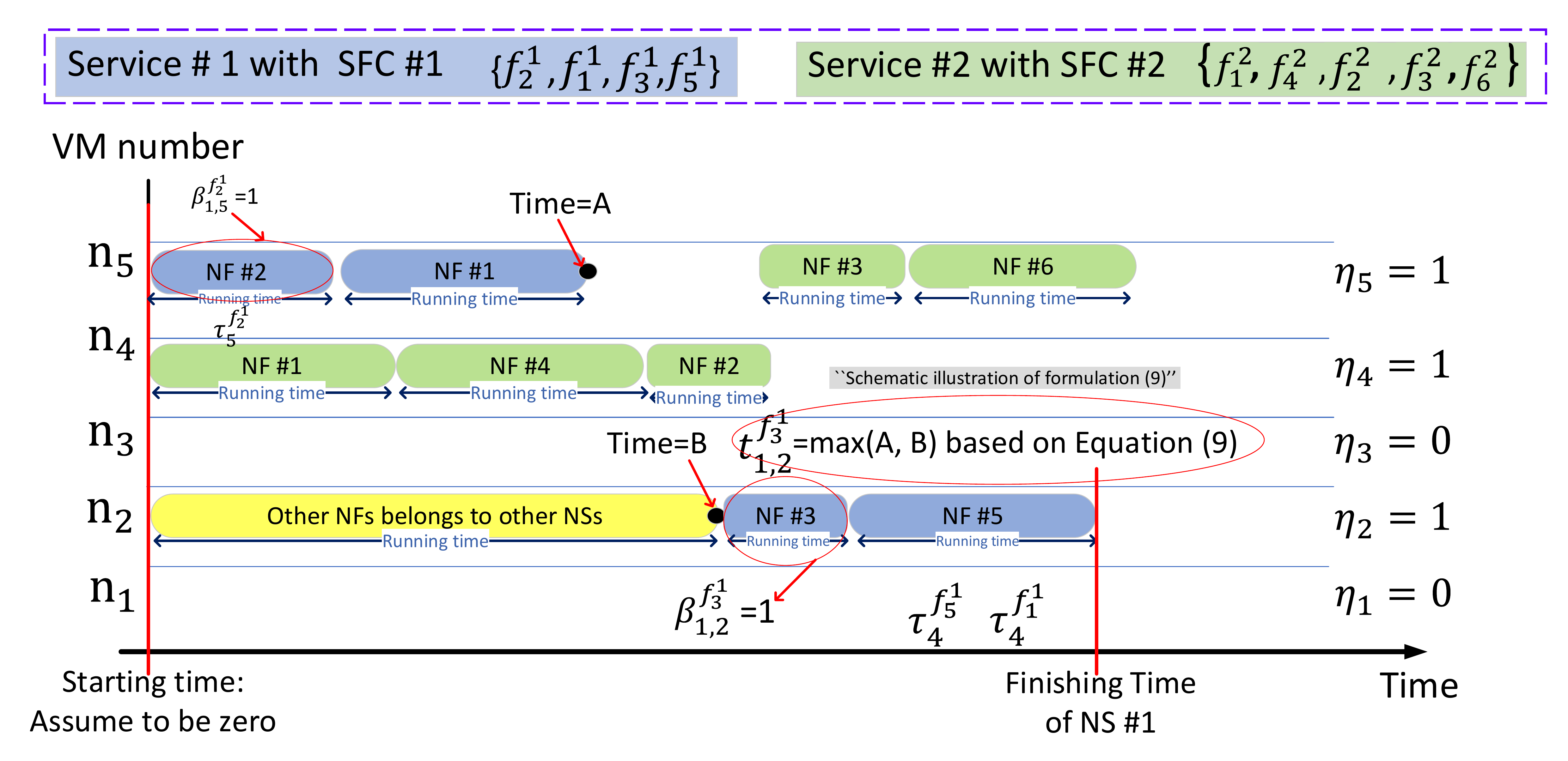}}
		\vspace{-1em}
		\caption{\textcolor{black}{Schematic illustration of the proposed scheduling and formulation of  \eqref{schule}.}}
		\label{Scheduling}
	\end{figure*}
	
	{It is worth noting that our problem is performed for a snapshot assuming all the packets of the services which are generated in unit time are fetched into the network at the beginning of each unit time.
	%
		  Hence, the arrival time of all packets is the same and can be set to zero.
		   Therefore, 
		the total service chain latency for each user $u$ on the requested service is inferred as follows \cite{game-theroy}:
	}
\begin{align}
&D^{\text{Total}}_{u}=\max_{\forall n\in\mathcal{N}, f_{m}^{s}\in\Omega_{s},s\in\mathcal{S}}\Big\{{t_{u,n}^{f_{m}^{s}}}\beta_{u,n}^{f_{m}^{s}}+\tau_{n}^{f_{m}^{s}}\beta_{u,n}^{f_{m}^{s}} \Big \}, 
\forall u \in \mathcal{U}.\label{maxd}
\end{align}
\subsection{\textbf{Cost Model: Objective Function}}
Our aim is to minimize the total cost of the network. In this regard, we define cost $\Psi$ as
the total amount of radio and NFV resources that are utilized in the network to provide services.
{In particular, the cost function is given as follows:  
	\begin{align}
	\Psi(\bold{P}, \boldsymbol{\rho}, \boldsymbol{\eta})=\mu_1 \sum_{u\in \mathcal{U},k\in \mathcal{K}}p_{u}^{k}+\mu_2\sum_{u\in \mathcal{U},k\in \mathcal{K}}B\cdot\rho_{u}^{k}+\mu_3\sum_{n\in \mathcal{N}}\eta_{n},
	\end{align}
	where  $\mu_1,\mu_2,\mu_3\ge0$ are constants with $\mu_1+\mu_2+\mu_3=1$ and are used for scaling and balancing the costs of different resource types. Notably, the units of these parameters are, respectively $ \$$/Watts, $ \$$/KHz, and $ \$$ for $\mu_1 $, $\mu_2$, and $\mu_3 $. Therefore, the unit of cost function is in $ \$ $.}
\subsection{\textbf{Problem Formulation}}  \label{Problemformulation}
Based on these definitions, our aim is to solve the following JRN-RA optimization problem: 
\begin{subequations}\label{1main}
	\begin{align}\label{op11}
		&\min_{\bold {P}, \boldsymbol{\rho},\bold{T},\boldsymbol{X}, \boldsymbol{\mathcal{\beta}}, \boldsymbol{\eta}}
	\Psi (\bold {P},\boldsymbol{\rho},\boldsymbol{\eta}),
	\\
		\textbf{ s.t:~~}
	\label{1maxr}
&	R_{u} \ge R^{\min}_{u},\, \forall u\in\mathcal{U} ,
	\\	&	\label{1rho}
	\sum_{u \in \mathcal{U}} \rho_{u}^{k} \le 1  ,\forall k\in\mathcal{K},
	\\&
	\label{1power}
	\sum_{k\in \mathcal{K}}  \sum_{u\in \mathcal{U}} \rho_{u}^{k}  p_{u}^{k}  \le P_{\max},
	\\&
	\label{1capcons}
	\sum_{u \in \mathcal{U}}	\sum_{s\in\mathcal{S}}\sum_{f_{m}^{s}\in\Omega_{s}}
	{y}_{u}\alpha^{f_{m}^{s}} \beta_{u,n}^{f_{m}^{s}} \le L_n, \forall n \in \mathcal{N}, 
	\\&
	\sum_{u \in \mathcal{U}}\sum_{s\in\mathcal{S}}\sum_{\forall f_{m}^{s}\in\Omega_{s}} \big(\psi^{f_{m}^{s}}  +
	{y}_{u}\big)\beta_{u,n}^{f_{m}^{s}}\le \Upsilon_n,\forall n \in \mathcal{N},\label{1bufcons}	 	
	\\&
	\label{1schul}
	t_{u,n}^{f_{m}^{s}}\beta_{u,n}^{f_{m}^{s}}\ge\nonumber\\&  \max \Bigg\{ \mathop {\max }\limits_{\scriptstyle\forall f_{m'}^{s'}\in\Omega_{s'},
		\scriptstyle u'\in\mathcal{U}} \left\{ {x_{u,u'}^{f_m^{s},f_{m'}^{s'}}\beta _{u',n}^{f_{m'}^{s'}}(t_{u',n}^{f_{m'}^{s'}} + \tau _n^{f_{m'}^{s'}})} \right\},
	\nonumber\\ &\mathop {\max }\limits_{\begin{array}{c}
	\scriptstyle{\forall f_{m''}^s\in \{{\Omega _s}-f_{m}^{s}\}}\\
		\scriptstyle  {n'\in \{\mathcal{N}-n\}}
		\end{array} } \left\{ {x_{u,u}^{f_{m}^{s},f_{m''}^{s}}\beta _{u,n'}^{f_{m''}^s}(t_{u,n'}^{f_{m''}^{s}} + \tau _{n'}^{f_{m''}^{s}})} \right\}\Bigg\},\nonumber\\&\forall f_m^s \in {\Omega _s}, f_{m'}^{s'}\in{\Omega_{s'}},{\mkern 1mu}  \forall s, s'\in\mathcal{S}, \forall n \in \mathcal{N},\, \forall u\in \mathcal{U},
	\\&
	\label{1maxd}
	D^{\text{Total}}_{u} \le D^{\text{max}}_{s},~\forall u\in\mathcal{U},
	\\&	 
	\label{1powerp}
	0\le p_{u}^{k},~\forall u\in\mathcal{U},~k\in\mathcal{K},
	\\  & 
	\label{1sercons}
	\beta_{u,n}^{f_{m}^{s}} \le \eta_{n},\, \forall n\in\mathcal{N},\forall u\in\mathcal{U}, f_{m}^{s} \in \Omega_{s},
	\\&
	\label{1maxfcons}
	\sum_{n \in \mathcal{N}}\beta_{u,n}^{f_{m}^{s}} \le 1, \forall u \in \mathcal{U}, f_{m}^{s} \in \Omega_{s}, s\in\mathcal{S},
	\\&
	\label{1rhod}
	\rho_{u}^{k}\in \{0,1\},~\forall u\in\mathcal{U},~k\in\mathcal{K},\\&
	\label{1betad}
	\beta_{u,n}^{f_{m}^{s}} \in \{0,1\},\forall u\in\mathcal{U},\,\forall f_{m}^{s}\in\Omega_{s},\forall s\in\mathcal{S},
	\\& 
	\label{1xd}
	x_{u,u'}^{f_{m}^{s},f_{m'}^{s'}} \in \{0,1\},\forall u,u'\in\mathcal{U},u\neq u', \forall f_{m}^{s},\forall f_{m'}^{s'}\in\Omega_{s'},
	\\& 
	\label{1etad}
	\eta_n \in \{0,1\},\forall n\in\mathcal{N},
	\end{align}
\end{subequations}
where $\boldsymbol{\rho}=[\rho_{u}^{k}]$, $\boldsymbol{\mathcal{\beta}}=[\beta_{u,n}^{f_{m}^{s}}]$, $\bold{T}=[t_{u,n}^{f_{m}^{s}}]$, $\bold{{P}}=[p_{u}^{k}], $  $\boldsymbol{X}=[x_{u,u'}^{f_{m}^{s},f_{m'}^{s'}}]$, and  $\boldsymbol{\eta}=[\eta_{n}] $. In problem \eqref{1main},  constraint 
\eqref{1maxr} ensures the minimum rate requirement, \eqref{1rho}  guarantees that each subcarrier is assigned \textcolor{black}{to at most one user}, \eqref{1power} is \textcolor{black}{the} transmit power constraint. Moreover, constraint \eqref{1schul} \textcolor{black}{determines the} scheduling  principle, constraint \eqref{1capcons}  guarantees the processing requirement for each NF with the corresponding packet size to run on the server, and \eqref{1bufcons} indicates \textcolor{black}{the} storage capacity requirement for both buffering and running NFs. \textcolor{black}{Constraints \eqref{1rhod}-\eqref{1etad} are  for binary variables.}
\section{\textbf{solution algorithm}}\label{solutions}
Optimization problem \eqref{1main} is non-convex including both mixed binary and continues variables with non-linear and non-convex constraints. \textcolor{black}{
	Hence, it is an intractable optimization problem obtaining whose optimal solution  requires high computational complexity and {time}
%
\cite{qu2016delay, nejad2018vspace}}. {Therefore, we intend to develop an algorithm to reach a polynomial order of complexity with local optimum. The adopted algorithm is based on the iterative decomposition method.
}

Without considering NFV-RA, 
\textcolor{black}{the radio RA problem, separately, on the power and subcarrier allocation variables is convex optimization problem, and hence, each of them can be solved efficiently.} While NFV-RA is MINLP with \textcolor{black}{large number of} variables, i.e., $\bold{T},\boldsymbol{X}, \boldsymbol{\eta}, \boldsymbol{\mathcal{\beta}}$. 
These motivate us to develop a new low complexity heuristic algorithm to solve NFV-RA sub-problem that is stated with details in Algorithm \ref{ALG_NFV}.  

\textcolor{black}{Since the optimization problem \eqref{1main} \textcolor{black}{can} be infeasible,} we propose a novel AC algorithm (see Section \ref{admissionc}) based on elasticization method by introducing a new elastic variable. In order to briefly explain of the elasticization method\footnote{Further information can be found in [Section 6.1.4 \cite{chinneck2007feasibility}]\cite{ebrahimi2019joint}.}, assume that we have a constraint 
$g({\bold{y}})\le {0}$, where $\bold{y}\in\Bbb{R}^{n}$ is the objective variable. We elasticize it by 
$g({\bold{y}})\le{A}$, where ${A}\ge{0}$ is the objective variable. Based on this method, the constraints that \textcolor{black}{would make the} JRN-RA optimization problem infeasible are changed as follows.
\textcolor{black}{By applying this method, we reformulate \eqref{1main}  as follows:}
\begin{subequations}\label{mainelas}
	\begin{align}
	\min_{\bold {P}, \boldsymbol{\rho},\bold{T},\boldsymbol{X}, \boldsymbol{\eta}, \boldsymbol{\mathcal{\beta}},{A}}
	&\Psi(\bold {P},\boldsymbol{\rho},\boldsymbol{\eta})+ W\cdot A,
	\\
	\textbf{ s.t:~}~&\label{elasrmax}
	R^{\min}_{u}-R_{u}\le A,\, \forall u\in\mathcal{U},\\		                         
	&\sum_{u \in \mathcal{U}}\sum_{s\in\mathcal{S}}\sum_{\forall f_{m}^{s}\in\Omega_{s}} [\psi^{f_{m}^{s}}+y_u]\cdot  \beta_{u,n}^{f_{m}^{s}}
	\nonumber\\&
	-\Upsilon_n\le A,\forall n \in \mathcal{N},\label{elasbuf}
	\\&
	\sum_{u \in \mathcal{U}}	\sum_{s\in\mathcal{S}}\sum_{f_{m}^{s}\in\Omega_{s}}{y}_{u}\alpha^{f_{m}^{s}} \beta_{u,n}^{f_{m}^{s}}			
	-L_n\le A,
	\\&\label{elaslatency}
	D^{\text{Total}}_{u} -D^{\text{max}}_{s}\le A ,\forall u\in\mathcal{U},
	\\&\label{elasvar}
	A\ge 0,
	\\&
	\eqref{1rho},~	\eqref{1power},~\eqref{1sercons},~\eqref{1maxfcons},~	\eqref{1schul},~\eqref{1powerp}-\eqref{1etad},\nonumber
	\end{align} 
\end{subequations}
where 
$A$ is the elastic variable and $W$ is a large positive number, i.e., $W\gg1$. 
\textcolor{black}{Note that since $A$ can be any non-negative value, the optimization problem \eqref{mainelas} is feasible. }
By solving the \textcolor{black}{optimization} problem \eqref{mainelas}, the  infeasibility of the main optimization problem \eqref{1main} is determined. Therefore,  if \textcolor{black}{the} elastic variable \textcolor{black}{$A$} is positive, problem \eqref{1main} is infeasible.
\textcolor{black}{To overcome the infeasibility of problem \eqref{1main}}, we introduce a new AC method \textcolor{black}{to reject some services providing rooms for the remaining ones}. 
In fact, our proposed solution of problem \eqref{1main}, namely, elasticization-based AC with ASM (E-AC-ASM), has three main steps; 1) elasticization: in this step the constraints  which make  problem \eqref{1main}  infeasible, are elasticized  2) solving problem \eqref{mainelas} with adopting ASM (see Algorithm \ref{ASM_algorithm}); 3) performing AC: in this step, based on the value of elastic variable, when $A$ is positive, we perform AC. The summary of the mentioned steps can be followed in Fig. \ref{1main}.
\textcolor{black}{	The block diagram illustrating the details of  E-AC-ASM to \textcolor{black}{solve the optimization} problem \eqref{1main} 
	is shown in Fig. \ref{SAD}.}
\subsection{\textbf{Admission Control}}\label{admissionc}
Our proposed AC is based on the value of elastic variable of problem \eqref{mainelas}. Whereas, if  $A$ is non-zero, the original problem \eqref{1main} is infeasible. This means \textcolor{black}{that} one or more elasticated  \textcolor{black}{constraints, i.e., \eqref{elasrmax}-\eqref{elaslatency}, are not satisfied.} 
To ensure these constraints, we can increase network resources (e.g., server's capacities) or reject some of \textcolor{black}{the users service requests. Since} the first method is not practical in more cases, we propose to reject some requested services by adopting the proposed AC. \textcolor{black}{ One of the major question in devising AC is} which one of the requested services \textcolor{black}{should be} rejected. In this case, the requested services have diverse characteristics and different \textcolor{black}{effects} on the \textcolor{black}{utilization of the} network resources, and consequently on the  infeasibility of problem \eqref{1main}. \textcolor{black}{To find the user which has the most effects on the infeasibility and reject its service, we 
	do as follows:}
\begin{align}\label{ADP}
& u^{\star}= \text{argmax}_{u} 
\Gamma_{u}\triangleq\kappa_1\Big|R^{\min}_{s}-R_u\Big|
+\nonumber\\&
\kappa_2\Big|\sum_{n \in \mathcal{N}}\Big(\sum_{s\in\mathcal{S}}\sum_{\forall f_{m}^{s}\in\Omega_{s}} [\psi^{f_{m}^{s}}  +
{y}_{u}]\cdot \beta_{u,n}^{f_{m}^{s}} 
-\Upsilon_n\Big)\Big|
\nonumber\\&+\kappa_3
\Big|	\sum_{n\in\mathcal{N}}	\sum_{s\in\mathcal{S}}\sum_{f_{m}^{s}\in\Omega_{s}}{y}_{u}\alpha^{f_{m}^{s}} \beta_{u,n}^{f_{m}^{s}}-L_n\Big|,
\end{align}
{where $\kappa_1\ge 0$  per bps, $\kappa_2\ge 0$, per bit  and $\kappa_3\ge 0$ per CPU cycle per second are the fitting  parameters to balance $\Gamma_{u}$ with units $  $, respectively. We emphasize that \eqref{ADP} calculates the gap between configuration values (e.g., minimum date rate and VM capacity) and feasible values.
}
 Moreover, in \eqref{ADP} we use the values of the optimization variables of \eqref{mainelas} obtained by Algorithm \ref{ASM_algorithm}. Based on this, we reject user $u^\star$. Then, solve problem \eqref{mainelas} with $\mathcal{U}'=\mathcal{U}-\{u^\star\}$. We repeat this procedure until, we have $A=0$ in the solution of problem \eqref{mainelas}.

\textcolor{black}{The re-formulated problem \eqref{mainelas} is also non-convex and intractable. In this regard, we solve} \textcolor{black}{it} by dividing it into  three sub-problems \textcolor{black}{by} utilizing ASM. \textcolor{black}{The} first sub-problem is  power allocation and elasticization,  \textcolor{black}{the} second \textcolor{black}{one} is subcarrier allocation, and \textcolor{black}{the last one}  is NFV-RA.
	In fact, the first and second sub-problems are the radio RA sub-problem and it is stated in Section \ref{RRASP}.
	In the NFV-RA sub-problem, all \textcolor{black}{the} optimization variables are integer and the problem formulation and solution are presented in Section \ref{NFVRASP}.  More details of the iterative solution of optimization problem \eqref{mainelas} are stated in Algorithm \ref{ASM_algorithm}. Moreover, we investigate the E-AC-ASM algorithm \textcolor{black}{from different aspects, namely, complexity, convergence and performance, and compare it with other existing methods. In the next subsection, we explain the solution of the aforementioned sub-problems. }
\begin{figure*}[t]
	\centering
	\centerline{\includegraphics[width=0.95\textwidth]{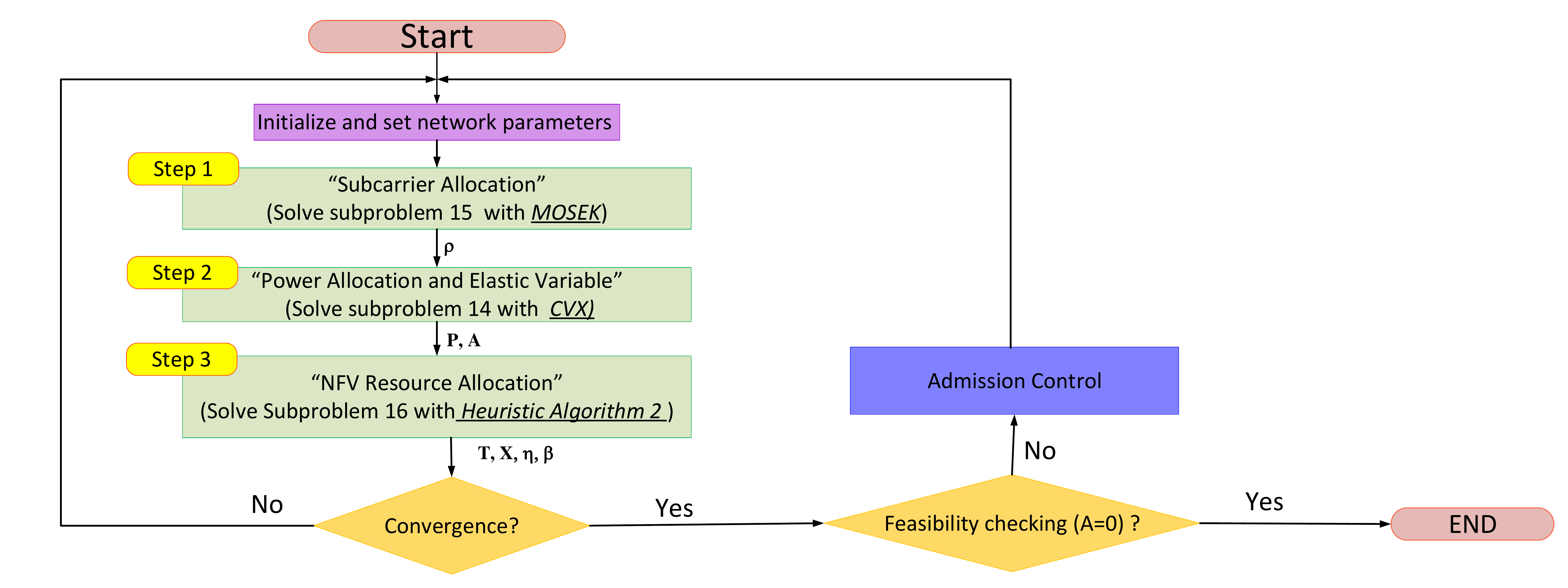}}
	\vspace{-1em}
	\caption{Flowchart of the E-AC-ASM algorithm  for solving  main problem \eqref{1main}.}
	\label{SAD}
\end{figure*}
\begin{algorithm}[t]
		\DontPrintSemicolon
		\label{ASM_algorithm}
	\renewcommand{\arraystretch}{0.9}
	\caption{ 
		Iterative  RA for solving problem \eqref{mainelas}}
	\KwInput{  $\epsilon_{\text{TH}}=10^{-4}$, $Z_{\text{TH}}=30$, $z=0$, $\bold{P}^{(z)}=[\frac{P_{\max}}{U\times K}]$,
		$\bold{T}^{(z)}=[t^0]$, $\boldsymbol{X}^{(z)}=[X^0]$, 
		$\boldsymbol{\beta}^{(z)}=[\beta^0]$, and $\boldsymbol{\eta}^{(z)}=[\eta^{0}]$}
	
		\Repeat{
				$|({\Psi}+WA)^{(z)}-({\Psi}+WA)^{(z-1)}|\le\epsilon_{\text{TH}}$ or  $z\geq Z_{\text{TH}}$}
		{
	\textbf{Step 1:} 	obtain subcarrier assignment variable, i.e., $\boldsymbol{\rho}$ by solving \eqref{submain}
		
	\textbf{Step 2:}	obtain power allocation variable, i.e., $\bold{P}$,  and elastic variable, i.e., $A$, by solving sub-problem \eqref{powerallocation}
		
	\textbf{Step 3:}	obtain NFV-RA variables, i.e., $ \bold{T},\boldsymbol{X}, \boldsymbol{\eta},$ and $\boldsymbol{\mathcal{\beta}}$
		by solving sub-problem \eqref{NFVP} by Algorithm \ref{ALG_NFV} 
		
		$z=z+1$
	}
\KwOutput{ $\boldsymbol{\rho}, \bold{P}, A, \bold{T},\boldsymbol{X}, \boldsymbol{\eta},$ and $\boldsymbol{\mathcal{\beta}}$}
\end{algorithm}

\vspace{-1em}
\subsection{\textbf{Radio Access Network  RA}}\label{RRASP}
The radio RA problem is divided into two sub-problems as follows.
\subsubsection{\textbf{Power Allocation and Elasticated Sub-problem}}
The power allocation and  elasticated sub-problem is presented as follows:
\begin{subequations}\label{powerallocation}
	\begin{align}
	\min_{\bold{P},A}&\sum_{u \in \mathcal{U}}\sum_{k \in \mathcal{K}}p_{u}^{k}+W\cdot A, 
	\\	\textbf{ s.t.:}~
\nonumber&\eqref{elasrmax}-\eqref{elasvar},~\eqref{1power},~\eqref{1powerp}.
	\end{align}
\end{subequations}
\textcolor{black}{Sub-problem} \eqref{powerallocation} is convex. Hence, it can be \textcolor{black}{solved} efficiently by using the interior point method (IPM) with CVX toolbox in MATLAB \cite{grant2006matlab}.
\subsubsection{\textbf{Subcarrier Allocation Sub-problem}}
The subcarrier allocation sub-problem is as follows:
\begin{subequations}\label{submain}
	\begin{align} 
	\min_{\boldsymbol{\rho}}& \sum_{u \in \mathcal{U}}\sum_{k \in \mathcal{K}}\rho_{u}^{k},
	\\\nonumber	\textbf{s.t:}~&\eqref{elasrmax},\,\eqref{1rho},\,\eqref{1power},\,\eqref{1rhod}. 
	\end{align}	
\end{subequations}
\textcolor{black}{Sub-problem \eqref{submain} is an integer linear programming problem, which can be solved by using binary convex optimization solver MOSEK \cite{mosek2015mosek}.}
\subsection{\textbf{NFV-RA}}\label{NFVRASP}
The NFV-RA sub-problem is as follows:
\begin{subequations}\label{NFVP}
	\begin{align}
	\min_{\bold{T},\boldsymbol{X}, \boldsymbol{\eta}, \boldsymbol{\mathcal{\beta}}}&
	\sum_{n \in \mathcal{N}} \eta_{n}
	\\
	\textbf{ s.t:}~~&\nonumber
	\eqref{elasbuf}-\eqref{elaslatency},~\eqref{1schul},~\eqref{1sercons}-\eqref{1etad}.
	\end{align}	
\end{subequations}
\textcolor{black}{To solve sub-problem \eqref{NFVP}, we propose a new  greedy-based algorithm as a heuristic algorithm
	because \eqref{NFVP} is non-convex with large number of variables.
		We map and schedule the functions on the servers to have the minimum processing latency based on the greedy criterion.}
	To this end, we ascendingly sort the servers by the total processing latency metric (greedy criterion).  After that, the server with the best rank, i.e., the highest available capacity in the sorted list, is turned on. Then, \textcolor{black}{either} we activate another server, if the previously activated servers cannot satisfy the resource demands by NFs or \textcolor{black}{we degrade the QoS of the users. Hence, our proposed algorithm is based on minimizing the number of active servers.}
Based on the algorithm, we ascendingly sort users according to latency requirements and then we start to map and schedule each of NFs on the servers. The details of the proposed NFV-RA are stated in Algorithm \ref{ALG_NFV}. 
\begin{algorithm} 
	\small
	\caption{Proposed greedy-based heuristic NFV-RA algorithm for solving sub-problem \eqref{NFVP}} 
	\label{ALG_NFV}
		\KwInput{ Set system configuration parameters: {$S,F$}, {$\Omega_{s}$}, {$\alpha^{f_{m}^{s}}$}, $\rho^{f_{m}^{s}}$} 
		
		  Sort in ascending order servers according to the total latency of NFs on these servers and write server's index in $\tilde{\mathcal{N}}$ (e.g., $\mathcal{N}=\{1,\dots,4\}\rightarrow$$\tilde{\mathcal{N}}=\{3,4,1,2\}$ for $N=4$)
		 
		 Sort in ascending order all users according to the value of $D_{u}^{\max}$ and write users's index in $\tilde{\mathcal{U}}$ (e.g.,  $\tilde{\mathcal{U}}=\{2,4,1,3,5\}$ for $U=5$)
		 	
		 \For {$\tilde{n}$=1: $\tilde{{N}}$ }
		{
		 Add $\tilde{\mathcal{N}}(\tilde{n})$ to set $\hat{\mathcal{N}}_{\text{Used}}$ (\textit{Servers are in this set are activated})
		
		 $\bold{t'}=[\bold{0}]_{1,\tilde{n}}$ (\textit{Servers in this set are activated}) \&
		 $\eta_{\tilde{n}}=1$
		 
		 \For { $\tilde{u}$=1: $\tilde{U}$}
		 {
		 \For  {$m$=1:$|\Omega_{s}|$}
		{
		 \If {$m \ge 2$}
		 {
		 \For {$\hat n$=1: $\hat{\mathcal{N}}_{\text{Used}}$  }
		 {
		 $\tilde {t}_{u,\tilde n}^{f_{m}^{s}}=\text{max}\{t_{u,n}^{f_{m-1}^{s}},t'(\tilde n)\}+\tau_{\tilde n}^{f_{m}^{s}}$
		}
	}
		 \Else
		 {
		%
		
		 \For {$\hat n$=1: $\hat{{N}}_{\text{Used}}$ }
		 {
		 $\tilde {t}_{u,\hat n}^{f_{m}^{s}}=t'(\hat n)+\tau_{\hat n}^{f_{m}^{s}}$
		}
		}
	
		%
		 Find server $n$ in set $\hat{\mathcal{N}}_{\text{Used}}$ which has lowest $\tilde {t}_{u,n}^{f_{m}^{s}}$ (greedy criteria)
		
		 $\beta_{u,n}^{f_{m}^{s}}=1$ \&		
		 Calculate $\phi_{u}^{f^{s}_{m}} =\sum_{n'\in\mathcal{N}}\beta_{u,n'}^{f_{m'}^{s}} (t_{u,n'}^{f_{m'}^{s}}+\tau_{n'}^{f_{m'}^{s}})$
		
		 \For  { $u' \in\mathcal{U}, u' \neq u$ \&  $f_{m'}^{s'} \in \Omega_{u',s'}$}
		 {
		 \If {$\beta_{u',n}^{f_{m}^{s'}}=1$ \& $\phi_{u}^{f^{s}_{m}} \ge t_{u',n}^{f_{m'}^{s'}}$}
		 {
		 $x_{u,u'}^{f_{m}^{s},f_{m'}^{s'}}=1$}
		}
		}
		 
		 Calculate $t_{u,n}^{f_{m}^{s}}$ according to constraint \eqref{schule}
		
		 $t'(n)=t_{u,n}^{f_{m}^{s}}$
	}

		 \If {\eqref{elasbuf}-\eqref{elaslatency} are satisfied} 
		 {
	     
	     	\textbf{Break}
	}
		 \Else    
		 { Return to line 5} 
		\vspace{-0.5em}
		}
		%
			\KwOutput{ $\boldsymbol{T}$, $\boldsymbol{X}$, $\boldsymbol{\eta}$, and $\boldsymbol{\beta}$}
\end{algorithm}
\section{\textbf{Convergence and Computational Complexity}} \label{Complexity-Convergance}
\subsection{\textbf{Convergence of the Solution Algorithm}}
\textcolor{black}{Based on ASM, after each iteration, the objective
function in each sub-problem is enhanced and finally it}
converges. Fig. \ref{ASM} shows an example \textcolor{black}{about the convergence of our proposed iterative algorithm. Clearly, it converges after  few iterations. } 	
\begin{figure}[t] 
	\centering
	\centerline{\includegraphics[width=0.42\textwidth]{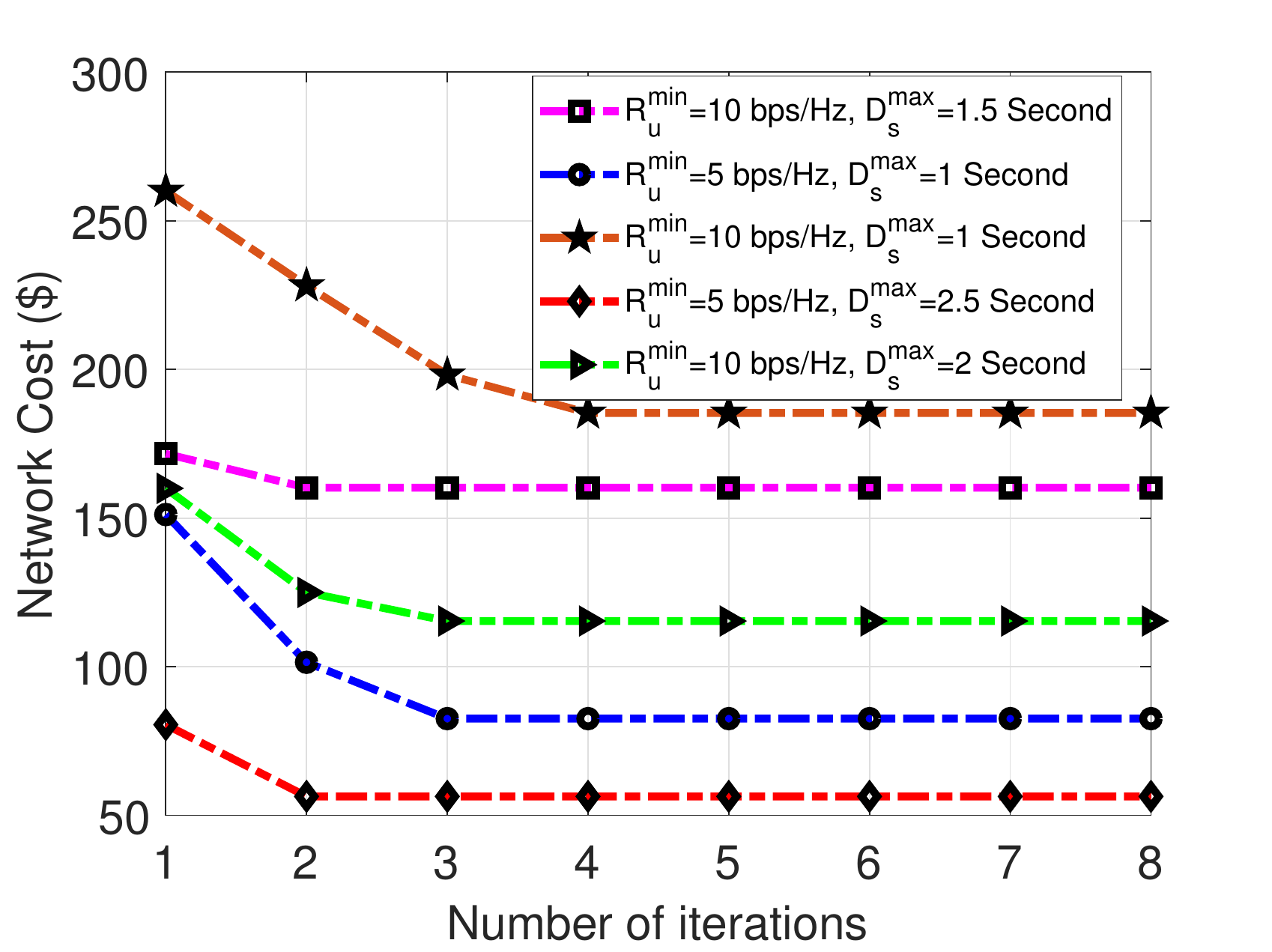}}
	\vspace{-1em}
	\caption{An examples of convergence of iterative solution with ($U=30, F=15, S=10, \text{and}~ N=20$) and other parameters are based on Table \ref{Core Network parameter setting} and \ref{setting}.}
	\vspace{-1em}
	\label{ASM}
\end{figure}
\begin{pro}
	With a feasible initialization of problem \eqref{mainelas}, the ASM algorithm converges to a \textcolor{black}{sub-optimal solution}.
\end{pro}
\begin{proof}
	{Please see Appendix A.}
\end{proof}
{Note that the value for the maximum iteration number, i.e., $ Z_{\text{TH}}$, in Algorithm 1 is considered from the algorithm implementation perspective  to ameliorate the run time and avoid the extra run time when a little improvement is achieved in the objective function. Therefore, from the theoretical perspective, the algorithm iterations are not limited. However, according to simulation results, the convergence of Algorithm 1 is obtained after few iterations.}
	\begin{pro}
	{	Algorithm  \ref{ALG_NFV} is a monotonic algorithm. Therefore, it generates the objective function values in the decreased order at each iteration $z$. Bear in mind that $ z $ is the iteration number of Algorithm \ref{ASM_algorithm}.}
	\end{pro}
	\begin{proof}Please see Appendix B.
		 \end{proof}
\subsection{\textbf{Computational Complexity}}
The main concerns behind developing an algorithm for solving optimization problems are the complexity order of the algorithm and the performance in terms of the optimality gap. To this end, we analyze the complexity of the proposed algorithm as well as the optimality gap which is discussed in Section \ref{Optimal_Gap}.  
	By utilizing the iterative approach, the overall complexity of the algorithm is a linear combination of the complexities of each sub-problem. Therefore, we discuss each algorithm and demonstrate that the order of the complexity is polynomial.
\subsubsection{\textbf{Radio RA}}
For the radio RA sub-problem, we utilize
geometric programming (GP) and  IPM via CVX  toolbox in MATLAB \cite{grant2006matlab}. Based on this method, the computational complexity order of power allocation sub-problem  is given by 
$\frac{\log(\frac{C_1}{\xi\varrho})}{\log(\varsigma)}$ where $C_1=2U+N(1+U)+1\approx N\times U,~U>N,
$ is the total number of constraints of sub-problem \eqref{powerallocation}, $\xi$ is the initial point for approximating the accuracy of IPM, $0<\varrho\ll 1$ is the stopping criterion for IPM, and $\varsigma$ is \textcolor{black}{ the accuracy of} IPM \cite{grant2006matlab}. Similarly, the complexity of sub-problem \eqref{submain} is given by $\frac{\log(\frac{C_2}{\xi\varrho})}{\log(\varsigma)}$ where $C_2=U+K+1$ is the total number of constraints of \eqref{submain}.
\subsubsection{\textbf{NFV-RA}}
Based on  the proposed heuristic algorithm in Algorithm \ref{ALG_NFV} for NFV-RA, \textcolor{black}{the complexity order of sub-problem \eqref{NFVP} is the total number of main calculations that are required for solving it.} Hence, the upper bound of complexity of Algorithm \ref{ALG_NFV} is $\mathcal{O}(U^2\times F\times N)$. The order of computational complexity of all sub-problems are summarized in Table \ref{com}.\\
%
\begin{table}
	\centering
	\renewcommand{\arraystretch}{1.2}
	\caption{Complexity order of the proposed solutions}
	\label{com}
	\begin{tabular}{|c|c|c}
		\hline
		\textbf{Algorithm}&	\textbf{Complexity} \\	
		\hline	
	Greedy-based heuristic algorithm & $\mathcal{O}(U^2\times F\times N)$
		\\
		\hline	
		Greedy-based algorithm & $\mathcal{O}(U^2\times F\times N)$
		\\
		\hline
		Power Allocation: CVX   & ${\log\left(\frac{C_1}{\xi\varrho}\right)}\big/{\log(\varsigma)}$ 
		\\
		\hline
		Subcarrier Allocation: CVX-MOSEK  &	${\log\left(\frac{C_2}{t^0\varrho}\right)}\big/{\log(\varsigma)}$\\
		\hline
	\end{tabular}
\end{table}
\section{\textbf{Experimental Evaluation}}\label{simulations}
In this section, we evaluate the proposed framework from different aspects and compare it with some baselines.  First of all, we present the network configuration ({next section}) and then discuss the obtained  results (Section \ref{Simulation_Results}).
\subsection{\textbf{Simulation Environment and Software Toolbox}}
	In this section, the simulation results are presented to evaluate the performance of the proposed system model. 
We consider $U=50$ users \textcolor{black}{which are} randomly distributed in the converge area of \textcolor{black}{a BS with radius} $500$ m, $\sigma=10^{-7}$ Watts, $h_{u}^{k}=x_{u,k}(d_u)^{-\varphi}$ where $\varphi=3$ is the path loss exponent, $x_{u,k}$ is the Rayleigh fading, and $d_u$ is the distance between the BS and user $u$ \cite{8786250}. Moreover, we set $K=64$, $S=20$, and $N=25$. 
We suppose that the users request services randomly with uniform distribution as  $\text{RS}_{u}\sim U_d[1,~ S]$.   We define $M=15$ different NFs with unique labels $1$-$15$, \textcolor{black}{i.e., $\mathcal{F}=\{f_1,\dots,f_{15}\}$}. Each NS $s$ is a combination of several NFs.    
Each NF in the requested service utilizes the existing network resources until its processing time is completed. 
{The radio network and scaling/fitting parameter settings are summarized in Table \ref{setting}. Please note that the scaling parameters are set based on the experimental evaluation of the objective function. Also in some cases such as pricing studies, some of them can be optimized, e.g., power unit cost in \cite{7899529}.}
\begin{table*}
	\centering
	\renewcommand{\arraystretch}{1.2}
	\caption{Radio access network configuration and fitting parameters values}
	\label{setting}
	\begin{tabular}{ |c|c|c|c|}
		\hline
		\textbf{Parameters(s)}  & \textbf{Value(s)}&\textbf{Parameter}  & \textbf{Value}  \\
		\hline
		$U$ & \text{Min}$=5$~~\text{Max}$=50$&$\mu_1$&1 $\$$/Watts
		\\
		\hline
		$K$ & $64$&$ \mu_2 $&1 $\$$/KHz
		\\
		\hline
		$\sigma$ &    $10^{-7}$~~\text{Watts}&$\mu_3$&10 $\$$
		\\
		\hline
		\text{BS radius} &    $500 ~\text{m}$&$\kappa_1$& 50/bps/Hz
		\\
		\hline
		$P_{\max}$ &    $40$ \text{Watts}&$\kappa_2$&1/MB
		\\
		\hline
		$\varrho$    &	$3$&$\kappa_3$&1/CPU cycle per second
		\\
		\hline
		$R^{\min}_{u}$    &	\text{Min}$=5$~~\text{Max}$=20$~~\text{in bps/Hz}&---&
		\\
		\hline
	\end{tabular}\label{simval}
\end{table*}
For \textcolor{black}{the} sake of clarity of the network configuration, also the main related core network  parameters  utilized in these simulations for creating the VMs and services are chosen randomly based on the uniform
distribution with the minimum and maximum values, i.e., $U_d[\min_{\text{Value}},~\max_{\text{Value}}]$ that are shown in
Table \ref{CNP} \cite{mijumbi2015design}.
\begin{table}[h]
	\centering
	\renewcommand{\arraystretch}{1.2}
	\caption{Core network configuration parameters values}
	\label{Core Network parameter setting}
	\begin{adjustbox}{width=.48\textwidth,center}
	\begin{tabular}{ |l|c|}
		\hline
		\textbf{Parameters(s)}  & \textbf{Value(s)}  \\
		\hline
		$N$ & \text{Min}$=15$~~\text{Max}$=40$~~
		\\
		\hline
		\text{Server storage/buffer capacity} &   \text{Min}$=1000$~~\text{Max}$=1500$ \text{MB}
		\\
		\hline
		\text{NF storage demand} &    \text{Min}$=5$~~\text{Max}$=15$ \text{MB}
		\\
		\hline
		\text{Number of services} &   \text{Min}$=10$~~\text{Max}$=25$
		\\
		\hline
		\text{Number of NFs}\text{ in each service} &   \text{Min}$=5$~~\text{Max}$=15$
		\\
		\hline
		Server processing capacity, i.e., $L_n$    &\text{Min}$=1500$~~\text{Max}$=3000$\\& \text{CPU cycle per second}
		\\
		\hline
		\text{Processing demand}\text{ of each NF}    &\text{Min}$=5$~~\text{Max}$=20$ \\&\text{CPU cycle per bit per unit time}
		\\
		\hline
		\text{Service processing deadline}    &\text{Min}$=0.3$~~\text{Max}$=7$ \text{ Second}
		\\
		\hline
	\end{tabular}\label{CNP}
\end{adjustbox}
\end{table}
The obtained results presented next are based on the simulation in Matlab Software and hardware with specs as Core i7 CPU and $8.00$ GB RAM.
\subsection{{Simulation Results}}\label{Simulation_Results}
	The simulation results are discussed in two main categories:
\\ 1) \textit{The investigation of the proposed system model under different network settings and parameters.} The results of this category are shown in Figures \ref{Acceptance_Ratio_Delay}-\ref{Server_Del}.
\\ 2) \textit{Comparison of the solution algorithm and framework with the considered baselines}. 
\\ 
{Note that the simulation results are obtained by averaging over 500 Monte-Carlo runs.}
 We discuss these results in the following.
\subsubsection{{Service Acceptance Ratio}}
The service acceptance ratio (SAR) is \textcolor{black}{defined} by the ratio of the number of accepted services by the network to the total \textcolor{black}{number of the requested services} by users and is \textcolor{black}{obtained by} $\varkappa=1-\frac{\hat{U}}{U}$ where  $\hat{U}$ is the number of users that their services are rejected based on the proposed AC. It is a criterion to investigate the efficiency of the proposed algorithm in utilizing total network resources to guarantee the requested QoS and accept the service demands.

As can be seen from {Figures \ref{Acceptance_Ratio_Delay}, \ref{Acceptance_Ratio_User}, and \ref{Acceptance_Ratio_Node}}, \textcolor{black}{the value of the} acceptance ratio depends on two main factors, i) the network resources capacity; ii) the number of users (service demands) and service QoS characteristics (latency and data rate).
Therefore, it is challenging to address high data rate and provide low latency services.

\begin{figure*}[t]
	\begin{center}$
		\begin{array}{cc}
		\subfigure[h][SAR versus the service deadline.]{
			\includegraphics[width=0.32\textwidth]{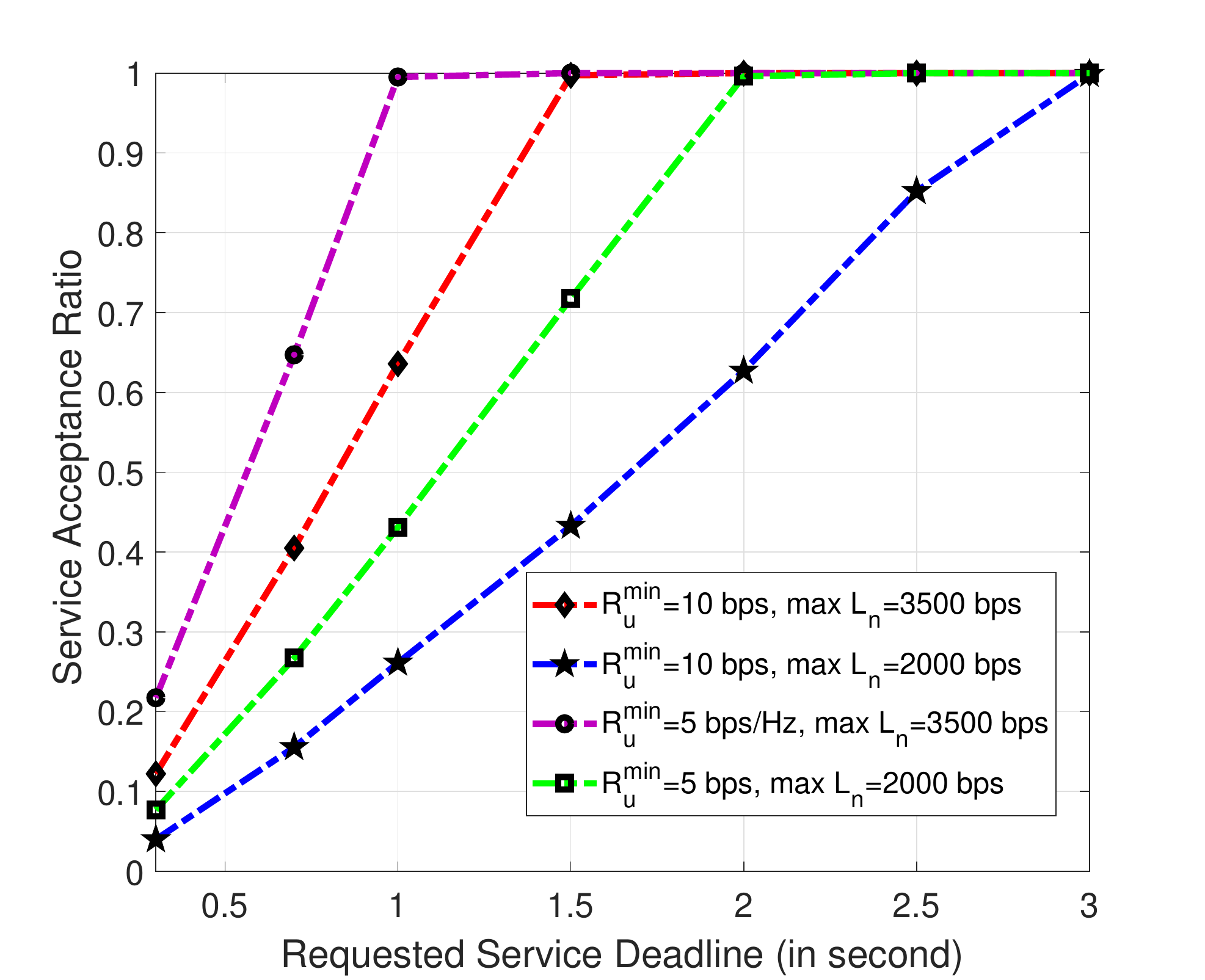}
			\label{Acceptance_Ratio_Delay}}
		\subfigure[SAR versus the total number of users.]
		{
			\includegraphics[width=0.32\textwidth]{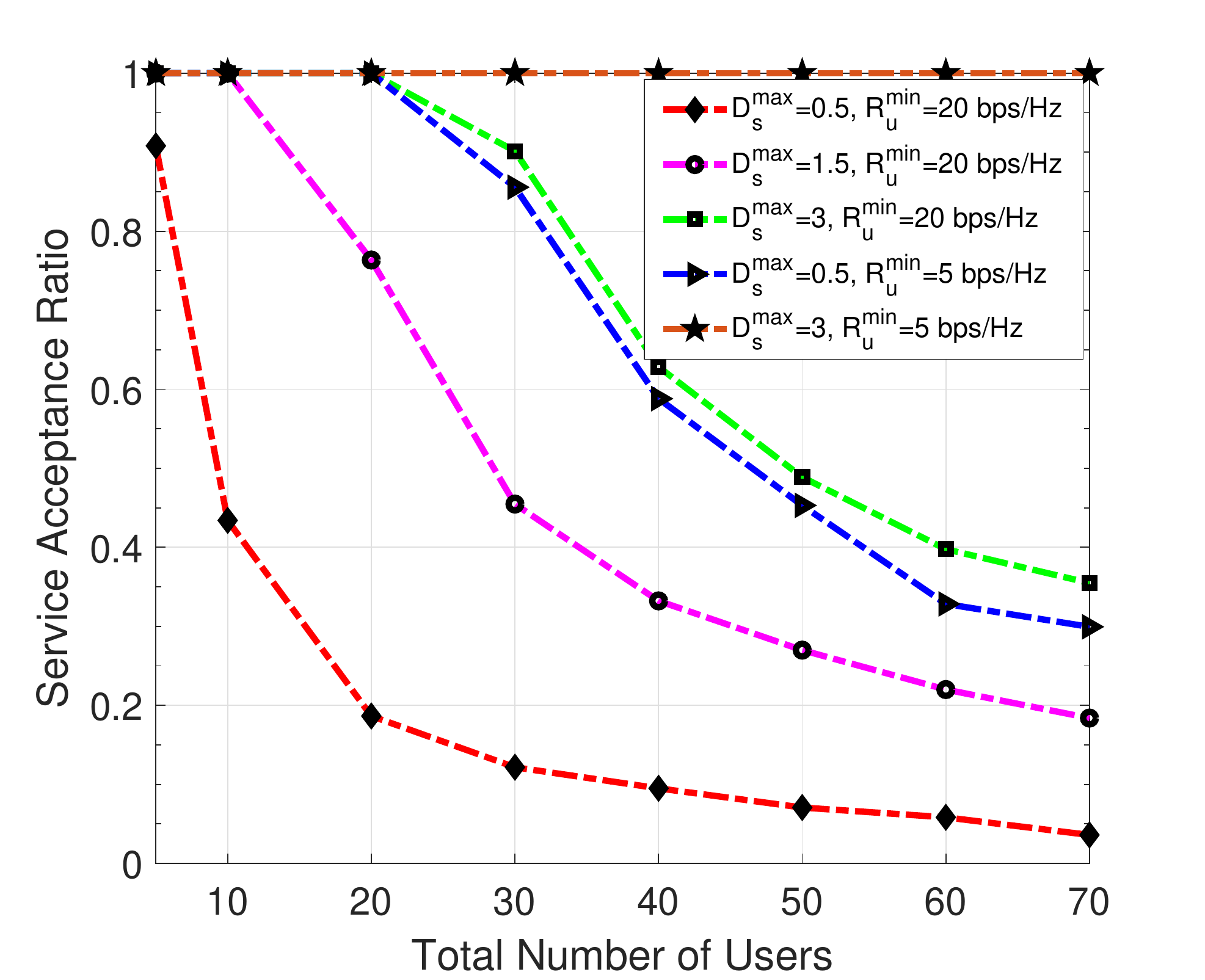}
			\label{Acceptance_Ratio_User}}
				\subfigure[h][SAR versus number of servers.]{
			\includegraphics[width=0.34\textwidth]{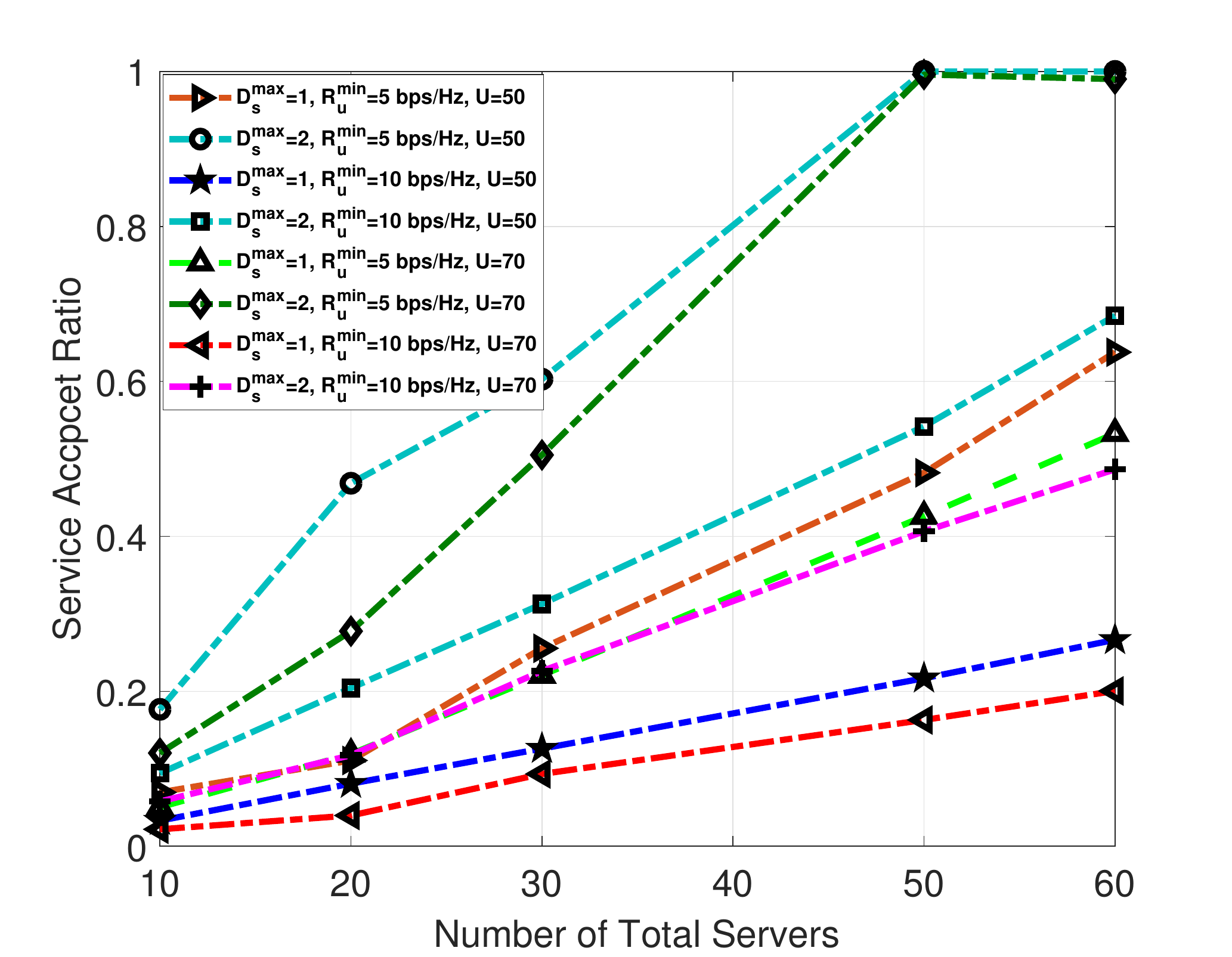}
			\label{Acceptance_Ratio_Node}}
		\end{array}$
	\end{center}
	\caption{Evaluation of SAR with different parameters.}
	\label{SAR}
\end{figure*}

{Fig. \ref{Acceptance_Ratio_Delay} shows the variation of \textcolor{black}{the SAR} with different \textcolor{black}{values of the} service latency and server's processing capacity. This result is obtained for $N=40$, $F=S=15$, $U=50$, $\alpha^{f_{m}^{s}}=20$, $P_{\max}=40$ Watts, and $R^{\min}_{s}= 5$ bps/Hz. \textcolor{black}{Form this figure, it can be seen that} the rejection probability of the low latency services is higher than that of \textcolor{black}{other types of} services.
	The reason for this is that these services need more servers with high processing capacity to reduce \textcolor{black}{the waiting and processing time.} It is clear that by increasing \textcolor{black}{the} latency from $0.3$ to $1$, the acceptance ratio is increased approximately \textcolor{black}{by} $50$\%.
	Moreover, this figure shows the impact of \textcolor{black}{the} minimum data rate on the acceptance of \textcolor{black}{the} service request. As can be seen from this figure, in contrast to the latency requirement, high data rate services are rejected by the network. 
	It would be better that we investigate the effect of latency versus the data rate. 
	Clearly,  if \textcolor{black}{the} minimum data rate value is doubled, on average approximately $\frac{0.4292}{0.7372}=0.58 \equiv 58$\% of users are rejected. 
	Whereas, if the latency is halved ($0.7$ to $0.3$), on average approximately $\frac{0.1155}{0.1555}=0.74 \equiv 74$\% of users are rejected by the network.
	Moreover, from this figure, we observe that by increasing \textcolor{black}{the} maximum processing capacity form $2500$ to $3500$, the SAR improves \textcolor{black}{by approximately $2$  times.} In other words, physical resource capacity has a major effect on the acceptance of services by the network, especially for \textcolor{black}{the} average latency of about $[1~2]$.
	Whereas,  high order latency services are not sensitive to the value of the server's capacity.}

Fig. \ref{Acceptance_Ratio_User} illustrates the variation of the value of SAR with the number of users (service arrivals) for different service deadlines and data rates. In this figure, we set $K=64$, $P_{\max}=80$ Watts, $N=40$, $F=15$, $L_n\in [500 ~2000]$ CPU cycles per second, and $S=15$.  \textcolor{black}{Clearly, by increasing the number of users (service requests) the acceptance ratio is decreased, especially for low latency services that have the main contribution on the acceptance ratio.
	We observe that increasing the number of low latency services leads to reducing the acceptance ratio. For the large number of users, the network guarantees some users' service requirements and other users are rejected. For this cases, based on $\varkappa$, the value of $\hat{U}$ is increased and $ U-\hat{U}$ approximately reaches to a fixed value.}

Fig. \ref{Acceptance_Ratio_Node} shows the variation of \textcolor{black}{the} SAR with increasing the number of servers for different scenarios.
In this figure, we set $U=[50~ 70]$, $\max L_n=2500$ bps, $\alpha^{f^{s}_{m}}=20$, and $M=S=15$. Clearly, increasing the number of  servers in the network improves the SAR.
Due to the fact that increasing the number of servers reduces the waiting time of NFs to run in the mapped servers. \textcolor{black}{On the other hand, the probability of  the large  number of mapped NFs on each server is  low and hence, the waiting and processing times are reduced.}
Therefore, the latency and buffering requirements are satisfied and the acceptance ratio of services is improved. From this figure, we conclude that the impact of the number of active servers on \textcolor{black}{the high data rate and low latency services e.g., process automation \cite{osseiran2015manufacturing} is more than that of other services. }
Furthermore, by comparing Fig. \ref{Acceptance_Ratio_Node} and Fig \ref{Acceptance_Ratio_Delay}, we obtain that the effect of \textcolor{black}{the} server processing capacity is more considerable than the number of active servers on the low latency services. 	That means the low latency services are rejected by the network because their requirements need more resources in the network to reduce waiting and processing times. 
\subsubsection{{Network Cost}}
Fig. \ref{Cost_U} illustrates the network cost versus the variation of the number of users for $R^{\min}_{s}=10$ bps/Hz and service deadline $2$ second. 
The network cost \textcolor{black}{is comprised of} both radio and NFV resources costs in terms of power and spectrum consumption and utilizing servers in the network. It can be observed that by increasing the number of users \textcolor{black}{the} network cost increases \textcolor{black}{due to increase in both the radio} and NFV costs.
It is clear that by increasing the number of users the NFV cost increases rapidly \textcolor{black}{compared to the radio cost.}

Fig. \ref{Ser_Uti} investigates the impact of the total number of users in the network on the \textcolor{black}{utilization of resources} with different minimum \textcolor{black}{data rates (as a packet size) and service deadlines.} In this experiment, we restrict the number of servers to $80$ with the maximum processing capacity $L_n=3000$, $M=15$, $S=20$, and $\alpha^{f_{s}^{m}}=20$. 
We define the utilization ratio as $\text{Uti}_{\text{Ratio}}=\frac{r_U}{r_T}$ where $r_U$ is the amount of the resources utilized by the users and $r_T$ is the total server's resources. 
From this figure, we infer that not only the packet size has a direct effect on the utilization ratio, but also the service deadline has a major impact on this. 
This is due to the fact that a large packet size needs more storage and processing capacity and low service deadline needs minimum waiting and processing times.
Therefore, we should make active more servers and exploit their resources for low latency services.  Obviously, increasing the number of users \textcolor{black}{increases} the utilization ratio approximately \textcolor{black}{in a linear form.
	From the cost perspective, we can conclude that by increasing the utilization of network resources, the \textcolor{black}{network cost} is also increased, especially in terms of power consumption. }
\begin{figure*}[t]
	\begin{center}$
		\begin{array}{cc}
		\subfigure[h][The network cost versus number of users.]{
		\includegraphics[width=0.45\textwidth]{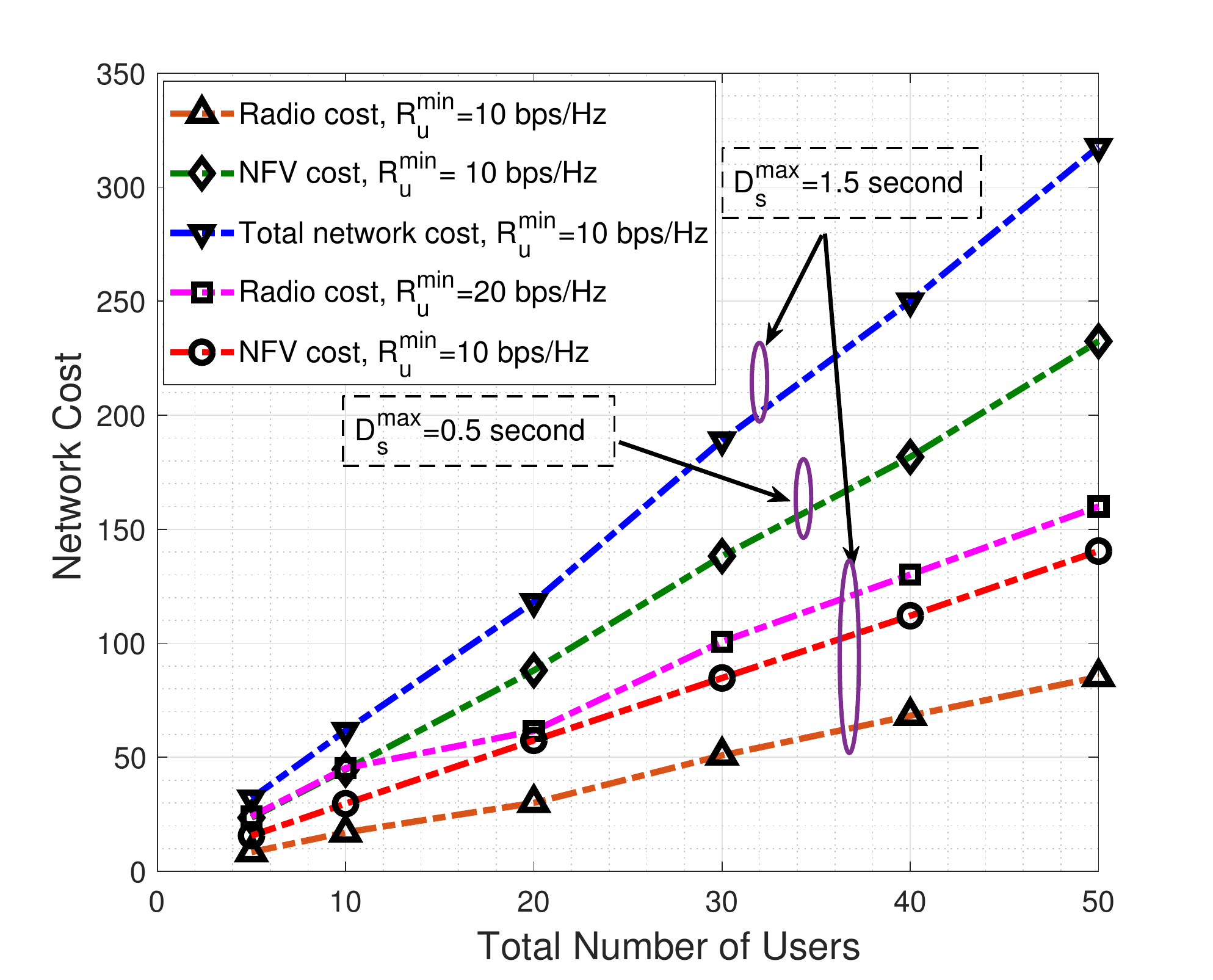}
			\label{Cost_U}}
		\subfigure[h][The network cost versus service deadline.]{
			\includegraphics[width=0.45\textwidth]{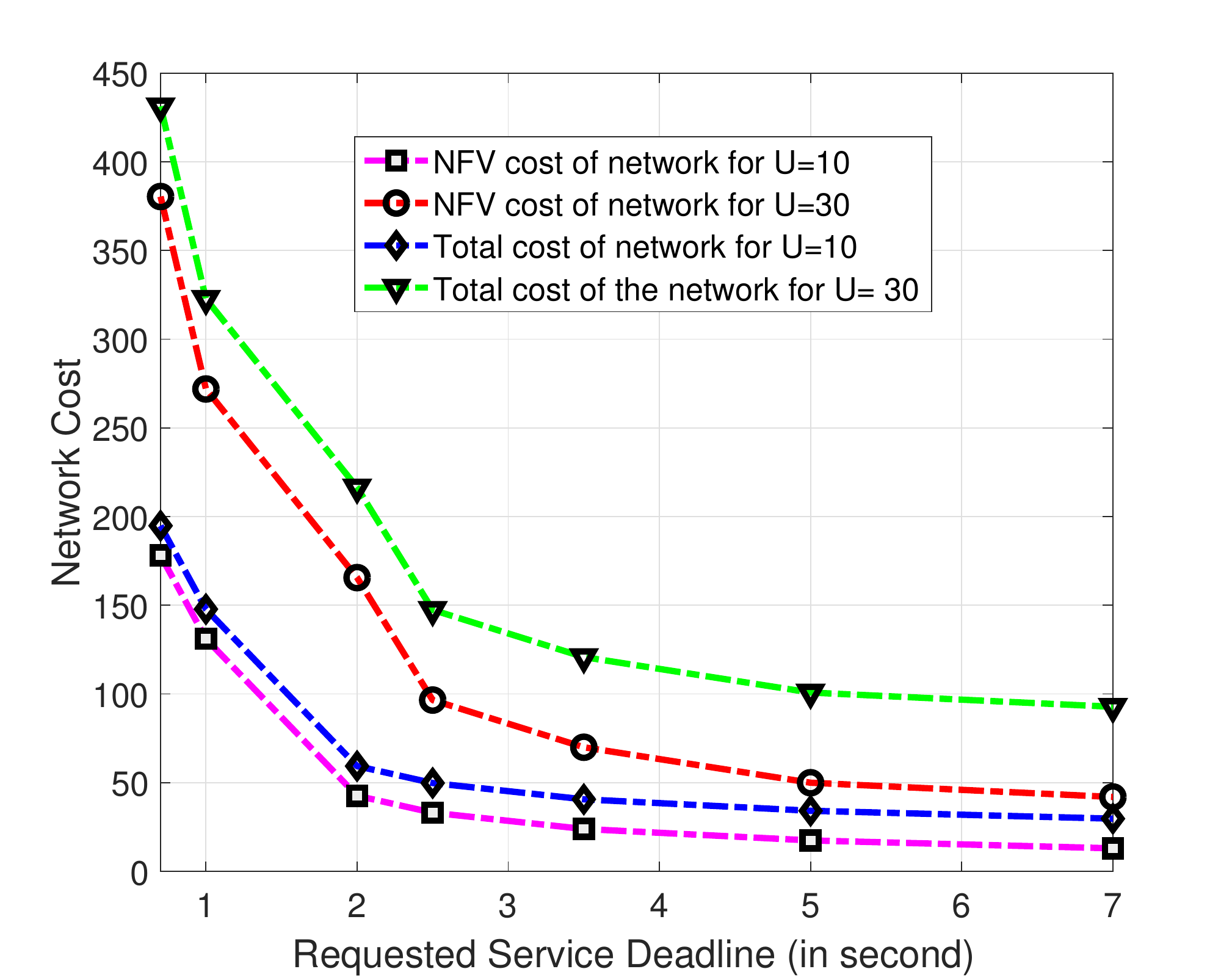}
			\label{Cost_Del}}
	\end{array}$
	\end{center}
	\caption{Evaluation of network cost with different parameters.}
	\label{NC}
	\end{figure*}
\subsubsection{{Service Deadline}}
\textcolor{black}{Figures \ref{Cost_Del} and \ref{Server_Del} show the total cost of the network versus different values of the service deadline} for various scenarios. 
Clearly, the requested service deadline has a major effect on the utilization of processing and buffering resources in servers. 
Form Fig. \ref{Server_Del}, we conclude that for services with lower latency requirements, more servers should be active to process the VNFs of the corresponding services.
That means for providing low latency services, we should pay more costs in terms of radio and NFV resources.
\textcolor{black}{By increasing the number of servers, the waiting time for each NF in a NS that it is in queue is minimized, and hence, server availability and probability of QoS guarantee for users are increased. }
For higher latency values in some cases, one (or two) active server(s) is sufficient.
By comparing Fig. \ref{Cost_Del} and Fig. \ref{Cost_U}, we obtain that by reducing \textcolor{black}{ the value of the latency}, the network cost \textcolor{black}{increases significantly compared to the case where the number of users (the numbers of service requested) increases.}
\begin{figure*}[t]
	\begin{center}$
		\begin{array}{cc}
		\subfigure[h][The ratio of the utilization of server's resources versus the \text{~~~~~} total number of users.]{
			\includegraphics[width=0.45\textwidth]{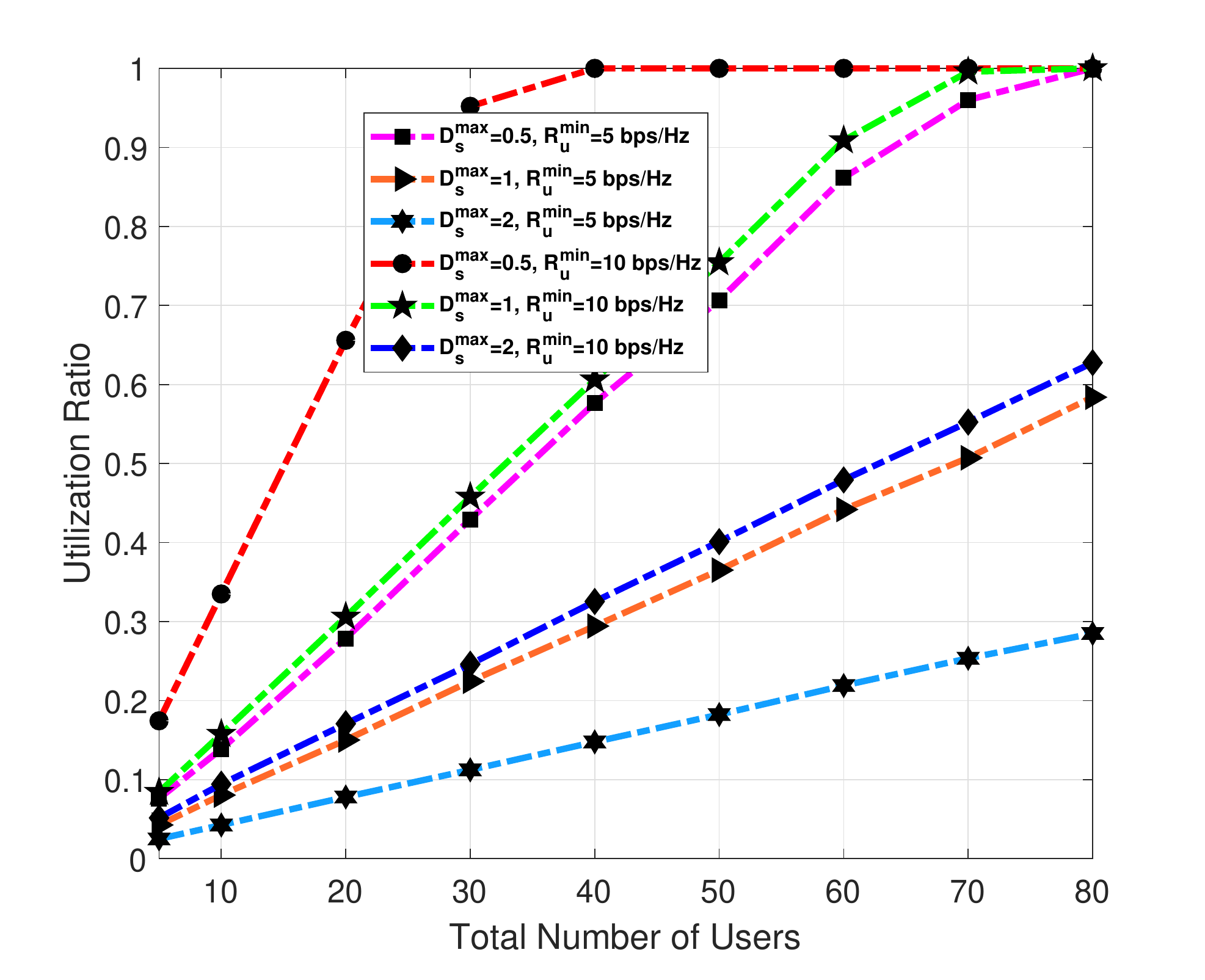}
			\label{Ser_Uti}}
		\subfigure[h][The number of active (on) servers versus the service deadline.]{
		\includegraphics[width=0.45\textwidth]{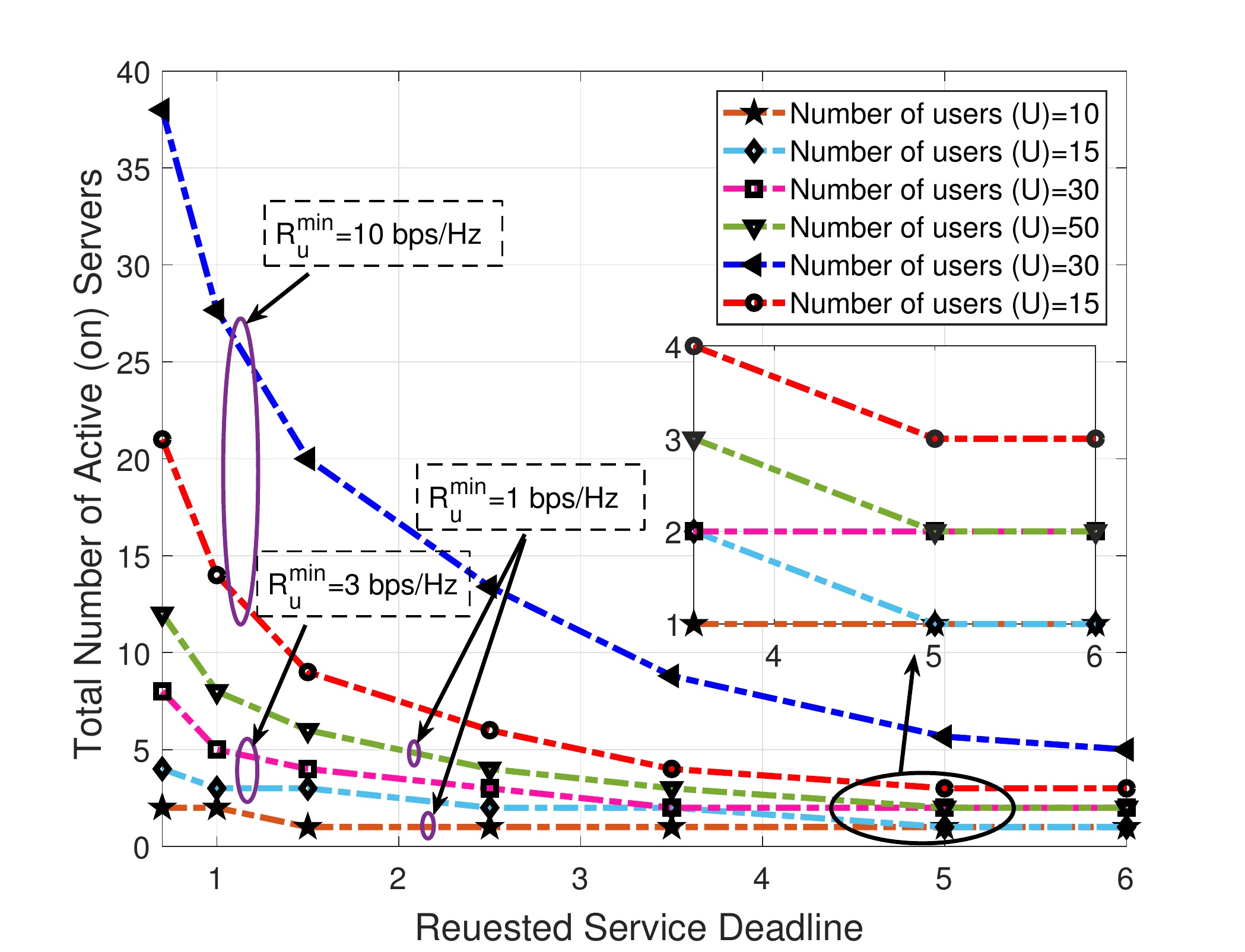}
			\label{Server_Del}}
		\end{array}$
	\end{center}
	\caption{Evaluation of utilization of resources in terms of the number of active servers and ratio of active servers's resources  with variation of different parameters.}
	\label{Ser_U_O}
\end{figure*}
\subsection{{Benchmark Algorithms }}\label{Bench_Mark}
To the best of our knowledge, this is the first work (refer to the related works) tackling the  effect of the radio resource on the NFV-RA. Moreover, we propose the new AC and a new closed form  \textcolor{black}{formulation of} NFV scheduling to minimize the total network cost. This is matched with the E2E network slicing concept. 
We compare our work with \cite{mijumbi2015design,8647504}, and \cite{7949048}  in terms of the proposed solution  algorithm design and framework. 	{Moreover, we compare performance of the proposed  iterative solution of original problem, i.e., (\ref{ASM_algorithm}) with global optimal for a small scaled network. Details are in the following.}
\subsubsection{Comparison with \cite{mijumbi2015design}}
 In \cite{mijumbi2015design}, the authors propose a greedy algorithm\footnote{{It is worth nothing that in the} related works, a greedy-based algorithm with different criteria is exploited \cite{mijumbi2015design, 6968961, zeng2018stochastic}.} [Algorithm 1 of \cite{mijumbi2015design}] for VNF placement and scheduling similar to our algorithm \ref{ALG_NFV}. Therefore, we compare our heuristic algorithm with a modified version of the greedy-based algorithm, which is proposed in \cite{mijumbi2015design}.
To avoid confusion  with our proposed greedy-based algorithm and the baseline greedy algorithm, we use the term heuristic/proposed algorithm for our algorithm.
	  In the greedy-based \textcolor{black}{ search,} different objectives can be considered, for example, \textcolor{black}{minimizing} the total flow time \cite{mijumbi2015design}.
The greedy-based scheduling and embedding of the arrived service requests are performed sequentially based on the greedy criteria. 
Based on the modified greedy algorithm to solve sub-problem \eqref{NFVP}, first, we search the servers that are appropriate for embedding and then find the best server by greedy criterion \cite{mijumbi2015design}. The steps of the greedy-based algorithm with the minimum queue time criterion is stated in Algorithm \ref{greedy} based on \cite{mijumbi2015design}.
\begin{algorithm}[!ht]
	\small
		\renewcommand{\arraystretch}{0.1}
	\caption{Greedy-based NFV-RA to solve sub-problem \eqref{NFVP} based on \cite{mijumbi2015design}}
	\label{greedy}
	\KwInput{ {$S,F$}, {$\Omega_{s}$}, {$\alpha^{f_{m}^{s}}$}, $\rho^{f_{m}^{s}}$, $t'(n)=0, \forall n\in\mathcal{N}$ is the last running time of server $n$  }
		 
		Sort in ascending order all users according to the latency requirement and write user's index in $\tilde{\mathcal{U}}_T$
		
		\For { $u$=1: $\tilde{U}_T$}
		{
		\For {$m$=1:$|\Omega_{s}|$ ($s$ is requsted service of user $u$)}
		{
			Check  processing and buffer constraints, i.e.,  \eqref{elasbuf} and \eqref{1rho}, and write servers that satisfy them in set of candidate servers as $\mathcal{N}_{\text{Can}}$
			
		 Sort in ascending order $\mathcal{N}_{\text{Can}}$ according to greedy criterion, i.e., the shortest queuing time for function $f_{m}^{s}$ 
		 
		Select the first rank server and set $\beta_{u,n}^{f_{m}^{s}}=1$ (index $n$ has first rank in $\mathcal{N}_{\text{Can}}$)
		
		\If {$m == 1$} {
				$ {t}_{u, n}^{f_{m}^{s}}=t'(n)$ }
		\Else
		{
		$ {t}_{u,n}^{f_{m}^{s}}=\text{max}\{t_{u,n}^{f_{m-1}^{s}},t'( n)\}$ \&
		$t'( n)=t_{u,n}^{f_{m}^{s}}+\tau_{ n}^{f_{m}^{s}}$
	}
		
		 Update the last released time of server $n$
		
		 Calculate $\phi_{u}^{f^{s}_{m}} =\sum_{n'\in\mathcal{N}}\beta_{u,n'}^{f_{m'}^{s}} (t_{u,n'}^{f_{m'}^{s}}+\tau_{n'}^{f_{m'}^{s}})$
		 
		\For  { $u' \in\mathcal{U}, u' \neq u$ \&  $f_{m'}^{s'} \in \Omega_{u',s'}$}
		{
		\If {
			$\beta_{u',n}^{f_{m}^{s'}}=1$ \& $\phi_{u}^{f^{s}_{m}} \ge t_{u',n}^{f_{m'}^{s'}}$}
		{
		$x_{u,u'}^{f_{m}^{s},f_{m'}^{s'}}=1$
	}
	}
	
		Calculate $t_{u,n}^{f_{m}^{s}}$ according to constraint \eqref{1schul} \&	 $t'(n)=t_{u,n}^{f_{m}^{s}}$
	}
}
		\KwOutput{ $\bold{T}$, $\bold{X}$, $\boldsymbol{\eta}$, and $\boldsymbol{\beta}$}
\end{algorithm}

Fig. \ref{Com_Accep_ratio} highlights the comparison of the proposed algorithm with the greedy algorithm \cite{mijumbi2015design} from the acceptance ratio perspective. As seen, the acceptance ratio of the heuristic/proposed algorithm is better than the greedy algorithm in \cite{mijumbi2015design}.
	For a small number of users, the results of both algorithms are the same. \textcolor{black}{ As a reason, since in this case resource requirements are low, both algorithms} \textcolor{black}{accept} almost all users.   
\begin{figure*}[!ht]
	\begin{center}$
		\begin{array}{cc}
		\subfigure[h][The SAR versus number of users.]{
			\includegraphics[width=.43\textwidth]{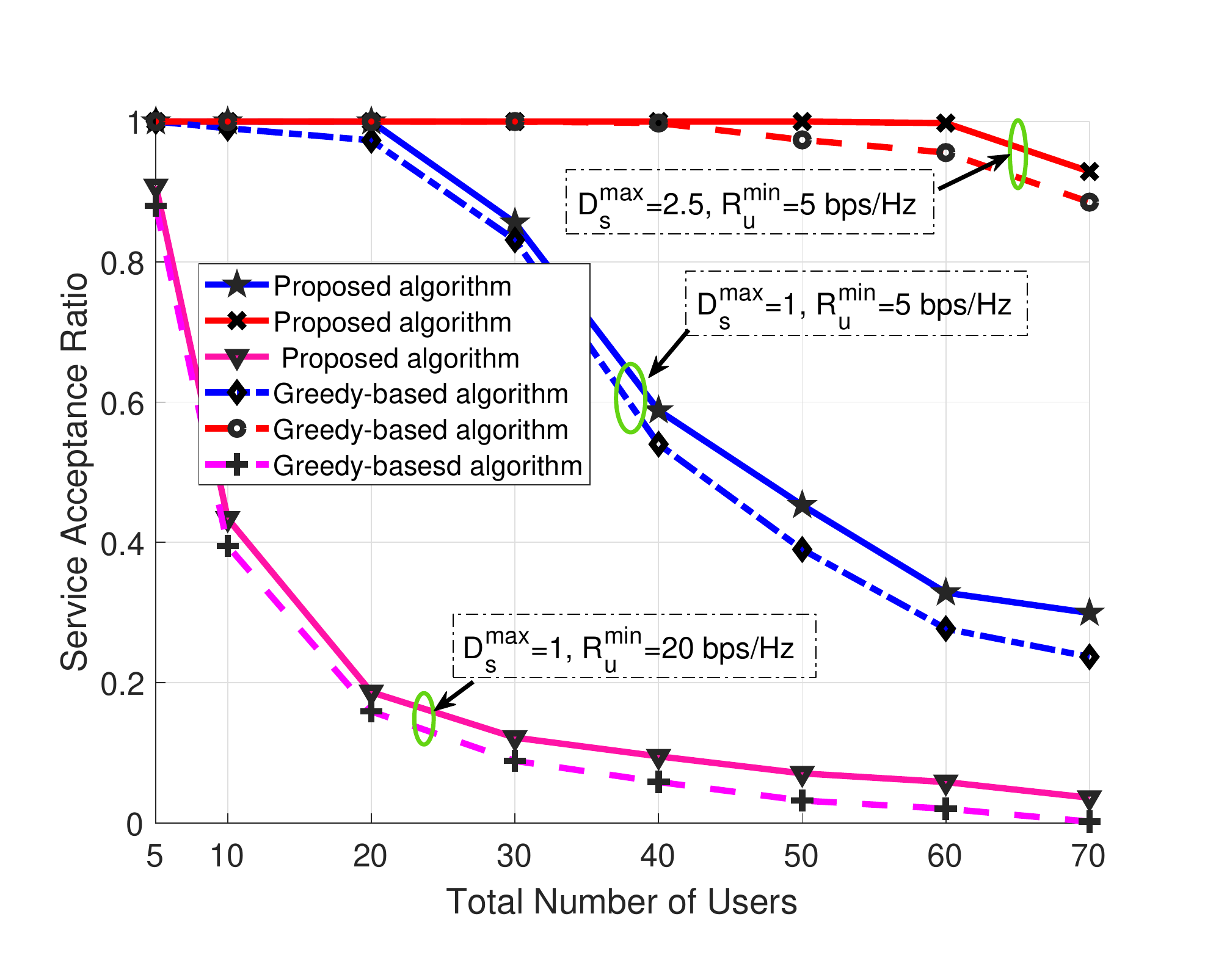}
			\label{Com_Accep_ratio}}
		\subfigure[h][Number of active servers versus the service deadline.]{
			{\includegraphics[width=.43\textwidth]{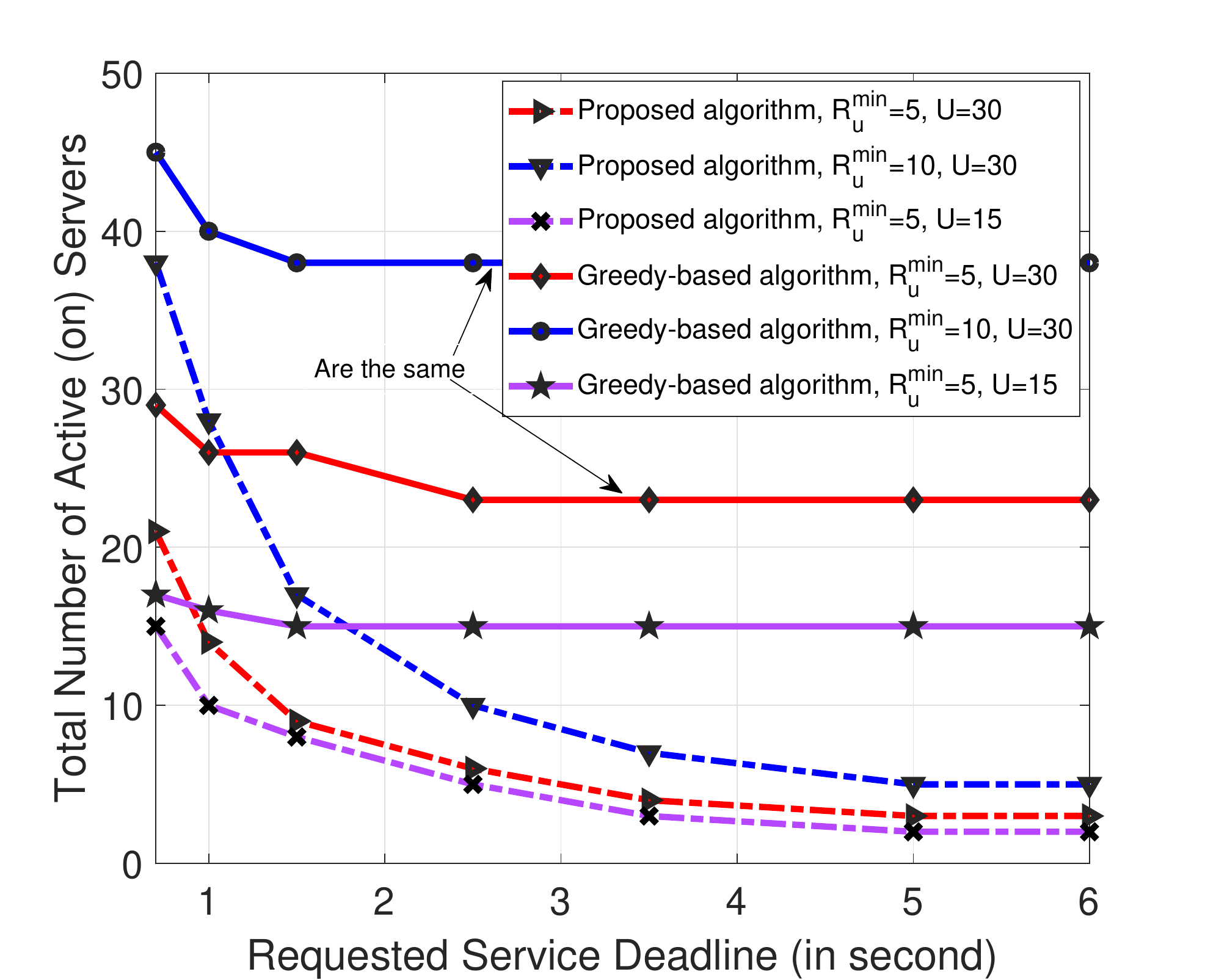}}
			\label{Com_Active_Server}}
	\end{array}$
	\end{center}
	\caption{\textcolor{black}{Comparison with the state of the art method {adopted} from \cite{mijumbi2015design}.}}
	\label{SCW}
	\end{figure*}
Moreover, we compare the impact of two mentioned algorithms on the number of active servers versus the requested service deadline in Fig. \ref{Com_Active_Server}. As seen, the number of active servers/VMs in our proposed algorithm is \textcolor{black}{lower than} that of the greedy algorithm.
As a reason, in the greedy algorithm for each NF, the algorithm finds a server with the lowest queuing time. In some cases, the algorithm adds servers \textcolor{black}{that are released} and have more processing capacities.  While it is possible to satisfy the latency of other functions without utilizing this server.  
In contrast, the heuristic algorithm \textcolor{black}{activates} a server when the \textcolor{black}{previously added servers (activated servers)} cannot satisfy the constraints of the problem and users QoS.  More importantly, in the greedy algorithm, the number of active servers is fixed after increasing the values of the service deadlines, which is the consequence of its server selection policy, which is based on the queuing time. However, as it could be seen,  in the heuristic algorithm, the number of active servers is reduced.

To \textcolor{black}{better demonstrate this, we assume that} we have five servers in the network with specific \textcolor{black}{capacities} as $[1000~ 2000~ 1500~ 3000~ 1800]$ and two \textcolor{black}{service requests} with $2$ functions with capacity requirements $20$ and $40$, $R^{\min}_{s}=10$ bps/Hz and service deadlines $0.3$ and $0.7$, respectively. 
Based on Algorithm \ref{ALG_NFV}, the service finishing time of user $1$ is $\frac{20\times 10}{3000}+\frac{40\times 10}{3000}=0.2<0.3$ and service finishing time of user $2$ is $0.2+ \frac{20\times 10}{3000}+\frac{40\times 10}{3000}=0.4<0.7$. That means one active server is sufficient \textcolor{black}{for} all users. 
While based  on the greedy algorithm, the finishing service time of user $1$ is $ \frac{20\times 10}{3000}+\frac{40\times 10}{3000}=0.2<0.3$ and \textcolor{black}{that} of user 2 is $\frac{20\times 10}{2000}+\frac{40\times 10}{2000}= 0.3<0.7$, since $0.3<0.4$, the greedy algorithm \textcolor{black}{ selects a server} with capacity $2000$ \textcolor{black}{instead} of the server with capacity of $3000$.
As a result, based on the greedy algorithm, two servers are utilized while in the proposed algorithm, \textcolor{black}{only}  one server in both cases \textcolor{black}{is utilized}.
Clearly, the greedy algorithm utilizes the servers inefficiently, and hence, \textcolor{black}{the acceptance ratio} is decreased especially for a large number of users (see Fig. \ref{Com_Accep_ratio}).
\subsubsection{Comparison With \cite{8647504}}
As mentioned before in the related work, the authors of \cite{8647504} aim to minimize the number of nodes hosting the NFs by considering the QoS and available resource constraints. We also follow this approach in NFV-RA by optimizing power and resource block/subcarrier allocation in the radio part. At the same time, it is reasonable to manage E2E resources in the context of the network slicing.  Therefore, we consider random and uniform power and subcarrier allocation, which can be considered for the radio part of \cite{8647504}. The curves, which represent the results of the comparison, are illustrated in Fig. \ref{SAR_COST_37}. As can be seen, our framework outperforms  the related literature in terms of both acceptance ratio and radio cost. Note that in this comparison, NFV-RA cost is ignored, and parameters are based on Table \ref{setting}.
 \begin{figure*}[!ht]
	\begin{center}$
		\begin{array}{cc}
		\subfigure[h][The acceptance ratio versus number of users.]{
			\includegraphics[width=.43\textwidth]{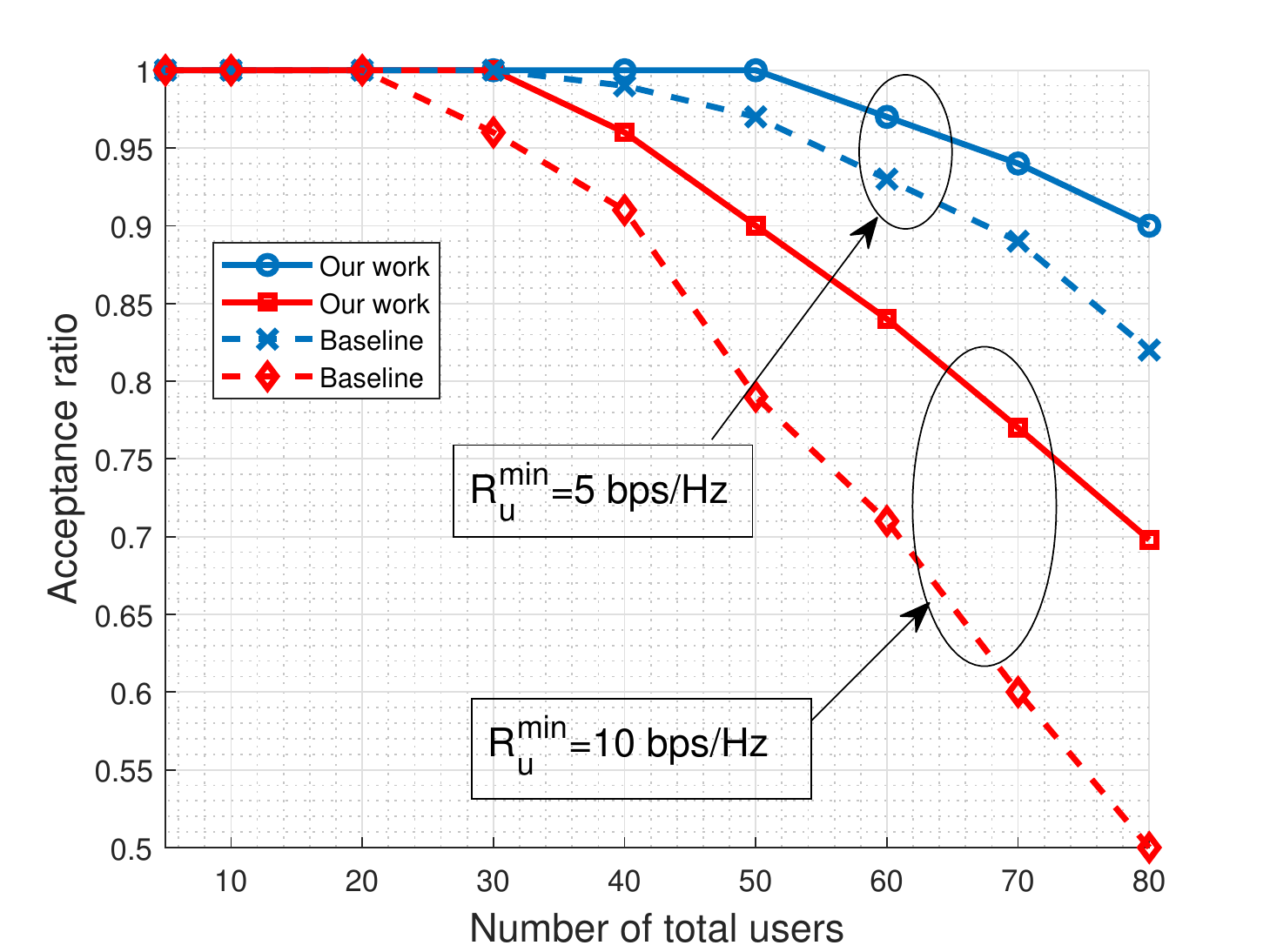}
			\label{Com_Accep_ratio_37}}
		\subfigure[h][Radio access network provisioning cost versus number of users.]{
			{\includegraphics[width=.43\textwidth]{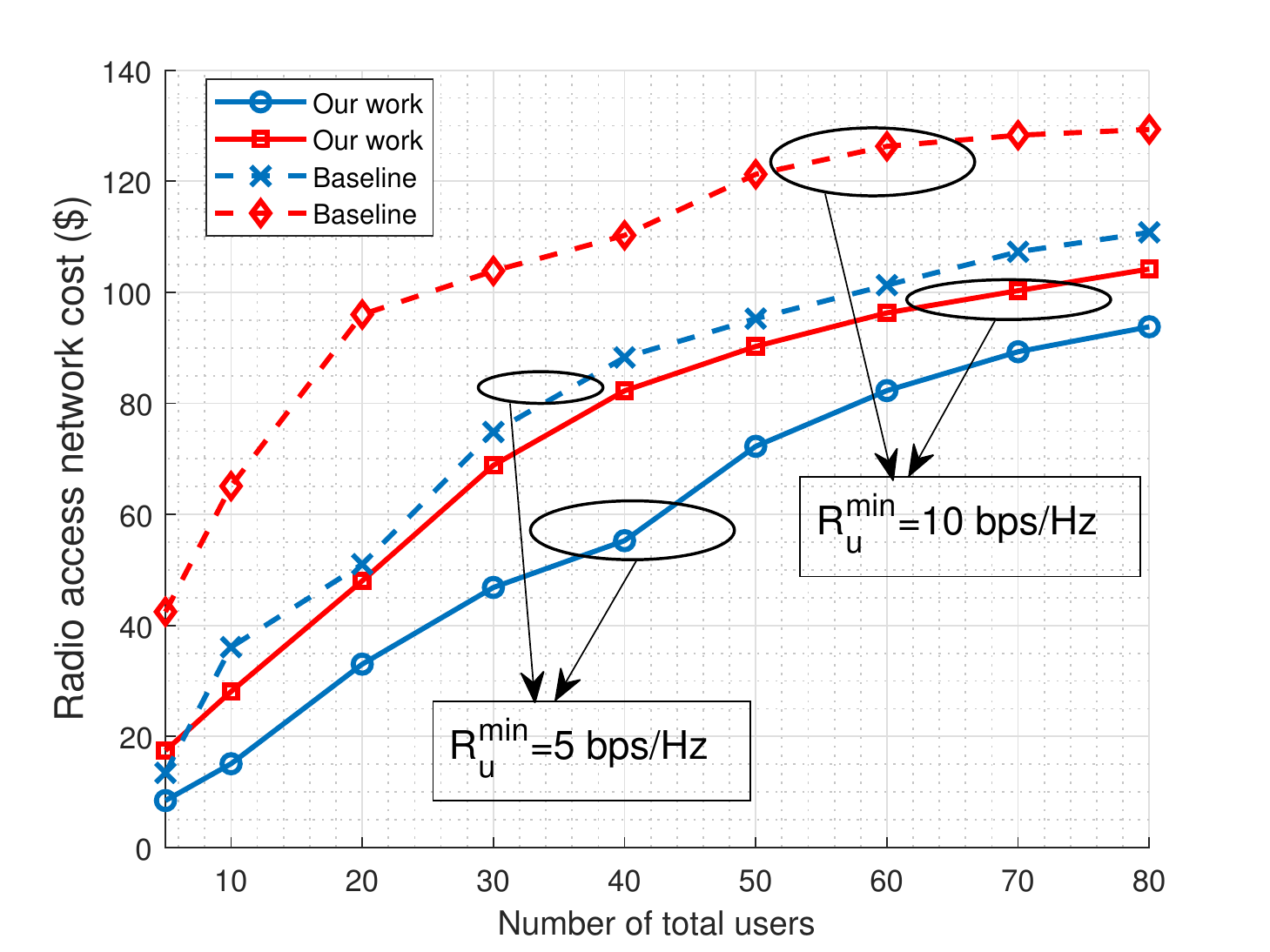}}
			\label{Com_Cost_37}}
		\end{array}$
	\end{center}
	\vspace{-1em}
	\caption{\textcolor{black}{Comparison with the state of the art framework/optimization {proposed} by  \cite{8647504}.}}
	\label{SAR_COST_37}
\end{figure*}
\subsubsection{Comparison With \cite{7949048}}
As stated in the related work, the authors in \cite{7949048} propose a power minimization problem for cloud-RAN to improve the energy efficiency
by optimizing the states of the computing units (CUs)  and radio access units (RAUs) (\textit{active or inactive}), RAU-user association, and  CU-RAU association\footnote{Where the VM of an RAU is assigned to appropriate CU based on optimization.}. 
However, the formulation of achieved data rate and user association are rather simplified, due to lack of consideration of the wireless channel and interference on SINR, and multiplication of the user association variable on date rate  [Eq. (4-6), \cite{7949048}]. 
Note that this simplification makes that the considered network and the optimization problem be impractical.
 At the same time, they assume that the transmission power of an RAU is equally allocated to each resource block similar to \cite{4786509}. Moreover, they do not consider the resource  block/subcarrier assignment problem and employ Max-SINR policy.  
Therefore, their user association is performed only based on the availability of the bandwidth on RAUs without considering the effect of channel condition on these resource blocks. 
    Moreover, our framework has some key differences compared to \cite{7949048}. By considering the objective of minimizing the energy as a cost, we compare our framework with \cite{7949048} in terms of optimizing power and subcarrier allocation. Fig. \ref{SAR_COST_41} gives the radio cost (Fig. \ref{Com_Cost_41}) and acceptance ratio (Fig. \ref{Com_Accep_ratio_41}) versus different number of users. It illustrates that the proposed framework performs better in terms of cost and user acceptance.  This is because of optimizing transmit power and subcarrier assignment in our scenario which have a pivotal role in wireless network performance and capacity \cite{8786250}, \cite{8654611}.
 \begin{figure*}[!ht]
 	\begin{center}$
 		\begin{array}{cc}
 		\subfigure[h][The acceptance ratio versus number of users.]{
 			\includegraphics[width=.43\textwidth]{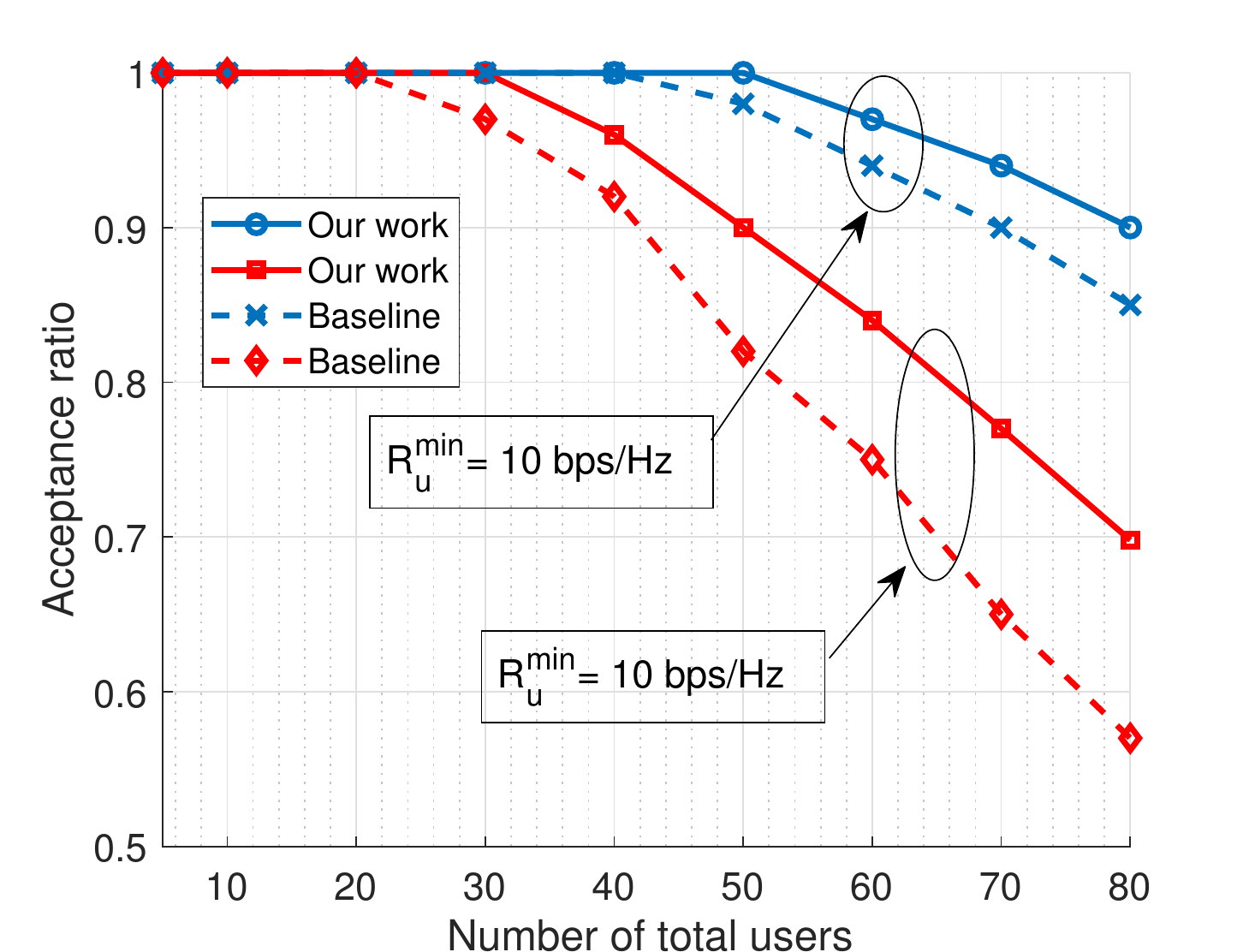}
 			\label{Com_Accep_ratio_41}}
 		\subfigure[h][Radio access network cost for provisioning versus number of users.]{
 			{\includegraphics[width=.43\textwidth]{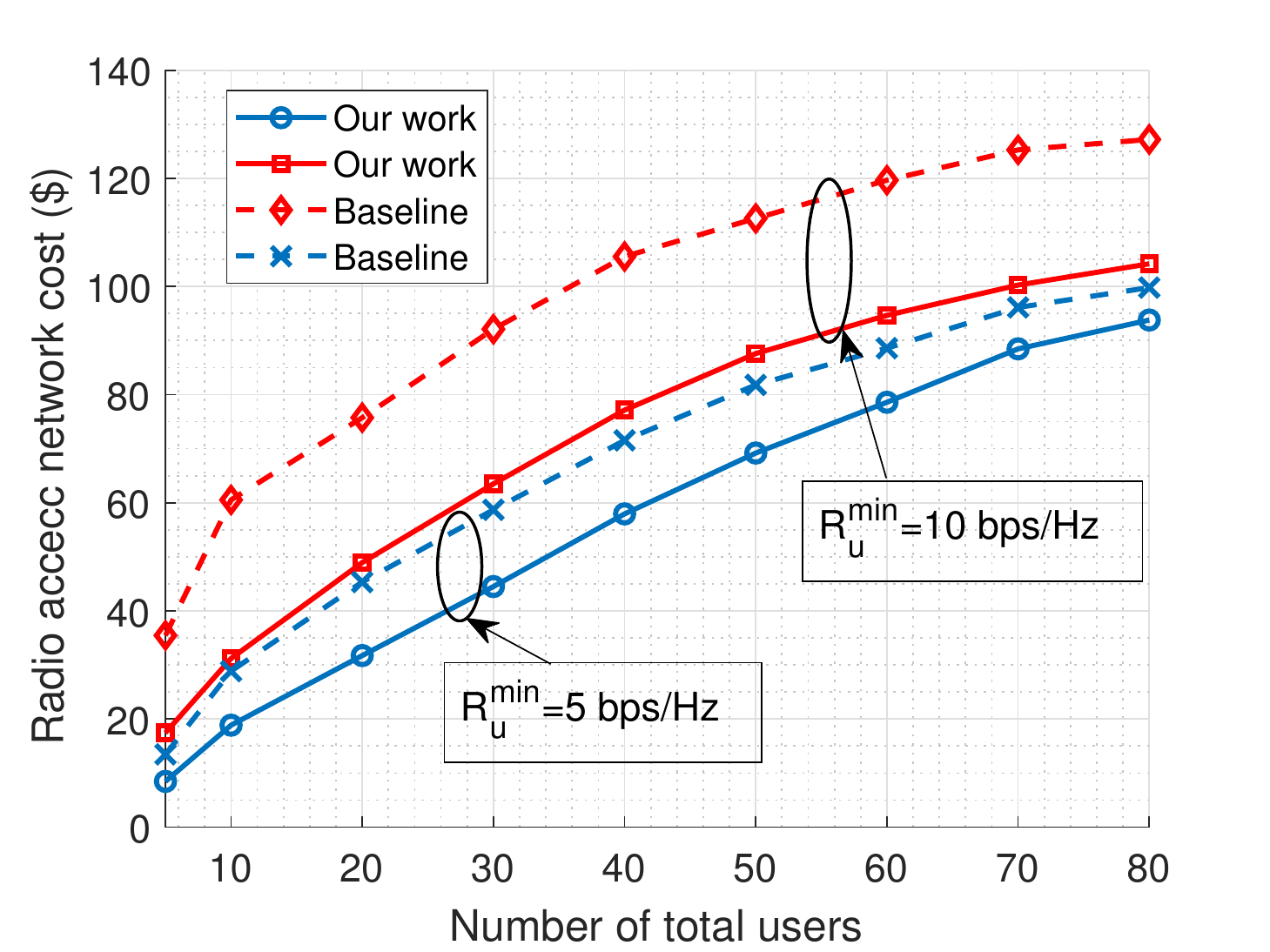}}
 			\label{Com_Cost_41}}
 	\end{array}$
 	\end{center}
 	\vspace{-1em}
 	\caption{\textcolor{black}{Comparison with the state of the art framework/optimization {proposed} by  \cite{7949048}.}}
 	\label{SAR_COST_41}
 	\end{figure*}
%
\subsubsection{{Optimality Gap}}\label{Optimal_Gap}
{Another baseline for investigation of the performance of the proposed solution algorithm is the optimality gap. In this regard,  we adopt the exhaustive search method \cite{8786250}. Since the complexity
of the exhaustive search method is very high and exponentially
grows with the size of  the system parameters, we exploit  it for a small scaled network. The considered
parameters and the corresponding solution methods
values are stated in Table \ref{Com_Al}. The  parameters are based
on Tables \ref{Core Network parameter setting} and \ref{setting}.
Results show that  our proposed algorithm gives approximately a $13.66 $\% optimality gap. 
	It can be readily noticed that the complexity of the proposed solution algorithm for the original problem is 
	 in the polynomial order compared to the exponential order for the optimal solution.  
\begin{table*}[!ht]
	\centering
	\small
	\caption{Performance comparison of different  solution algorithms} 
	\begin{tabular}{p{.1cm} p{.15cm}|p{1cm}|p{1.25cm}|p{1cm}|}
		\toprule 
		\midrule
		\multicolumn{5}{r}{Solution Methods}\\ 
		\cmidrule{3-5}
		&& ASM-Greedy  &  ASM-Proposed& Optimal\\ 
		\cmidrule{3-5}
		\multicolumn{1}{l}{\multirow{4}{*}{\begin{sideways}Scenarios\end{sideways}}}   &
		\multicolumn{1}{l}{Network Cost $\Big(M=S=5, N=10, U=K=10, D_{u}^{\max}=1.5, R^{\min}_{s}=10\Big)$}&254  &232  &217 \\
		\cmidrule{2-5}
		\multicolumn{1}{c}{}    &
		\multicolumn{1}{l}{SAR $\Big(M=S=5, U=40, N=15, K=10, D_{u}^{\max}=2, R^{\min}_{s}=15\Big)$}& 0.59 &0.64 & 0.7\\
		\cmidrule{2-5}
		\multicolumn{1}{c}{}    &	
		\multicolumn{1}{l}{Number of Activated server  $\Big(U=15, M=S=5, K=10, D_{u}^{\max}=1, R^{\min}_{s}=10\Big)$}& 16  &9&7 \\
		\midrule
	\end{tabular}
	\label{Com_Al}
\end{table*} 
\section{{Conclusion}}\label{conclusions}
In this paper, we proposed an E2E resource allocation and QoS assurance framework in NFV-enabled networks for heterogeneous services by realizing the network slicing paradigm.
	  This is achieved by formulating JRN-RA problem, in which the aim is to minimize the utilization of the radio resources in terms of power, spectrum, and the number of activated servers.
	  To solve the JRN-RA problem, we proposed the E-AC-ASM algorithm, where the elasticized problem is divided into three sub-problems, and then, each of them is solved efficiently. To solve NFV-RA, we proposed a low complexity greedy-based heuristic algorithm, which is based on minimizing the number of active servers in the network. 
	  By this scheme, we can reduce the resource \textcolor{black}{consumption}, such as processing, buffering, and power consumption belonging to each VM.

	We evaluated the performance of the proposed scheme with different parameters, such as QoS parameters, network resource capacities, and the performance metrics such as SAR and the number of active servers by numerical results. 
	 Moreover, to verify the performance of the proposed heuristic algorithm, we compared it with the state of the art schemes in terms of the number of the active servers and SAR (see Section \ref{Bench_Mark}). Our simulation results demonstrated that our solution algorithm outperforms the existing ones. By simulation and adopting the exhaustive search method, we investigated the optimality gap of the proposed iterative solution.  We showed that we can achieve about  $13.66 $\% optimality gap with polynomial order of complexity. 
 
 At the same time, the isolation guarantee between slices is challenging and becomes an interest research topic.  Since we studied the bit-level scheduling and E2E resource allocation, our frameworks are capable to isolate the traffic between different slices and flow of users.
%
 	 However, the comprehensive study in these areas focusing on automated orchestrator is planned for future work.
\appendices 
\section*{Appendix A}\label{Appendix_A}
\section*{Proof of Proposition 1}
Remind the objective function of problem (14)  as follows:
\begin{align}
&\Lambda(\bold{P},\boldsymbol{\rho}, \boldsymbol{\eta},A):=
\nonumber \\&
\underbrace{\mu_1\sum_{u\in \mathcal{U}} \sum_{k\in \mathcal{K}}\rho_{u}^{k}+\mu_1 \sum_{u\in \mathcal{U}} \sum_{k\in \mathcal{K}}p_{u}^{k}+\mu_3\sum_{n\in\mathcal{N}}\eta_{n}}_{\Psi}+W\cdot A. \nonumber
\end{align}
\textcolor{black}{We have the following relations between iterations ($z$ is the iteration number):}
\begin{align}\nonumber
&	\Lambda\Big(\bold{P}[z],\boldsymbol{\rho}[z],\boldsymbol{\eta}[z],A[z]\Big)=\min_{\bold{P},A}\Lambda\Big(\bold{P}[z],\boldsymbol{\rho}[z],\boldsymbol{\eta}[z],A[z]\Big)
\\\nonumber&
\le \Lambda\Big(\bold{P}[z-1],\boldsymbol{\rho}[z],\boldsymbol{\eta}[z],A[z-1]\Big)
\\\nonumber&
=\min_{\boldsymbol{\rho}}	\Lambda\Big(\bold{P}[z-1],\boldsymbol{\rho}[z],\boldsymbol{\eta}[z],A[z-1]\Big)
\\\nonumber &
\le  \Lambda\Big(\bold{P}[z-1],\boldsymbol{\rho}[z-1],\boldsymbol{\eta}[z],A[z-1]\Big)
\\\nonumber &
=\min_{\boldsymbol{\eta}} \Lambda\Big(\bold{P}[z-1],\boldsymbol{\rho}[z-1],\boldsymbol{\eta}[z],A[z-1]\Big) 
\\\nonumber&	\le	
\Lambda\Big(\bold{P}[z-1],\boldsymbol{\rho}[z-1],\boldsymbol{\eta}[z-1],A[z-1]\Big).
\end{align}
\textcolor{black}{	This means that the objective function of ASM decreases as the iteration number increases. There is also a lower bound (zero) and therefore there must exist a convergent sequence.}  In addition, with  QoS and ensuring the resource demand  constraints, i.e., (14b)-(14e), 
the ASM algorithm converges to a sub-optimal solution which corresponds to the sub-optimal solution of problem (14).
\section*{Appendix B}\label{Appendix_B}
\section*{Proof of Proposition 2}
Algorithm 2  works based on the values of  QoS metrics ($D^{\text{max}}_{s}$ and $R_{s}$) and capacity requirement of  NFs that are in the requested NSs. For the given system parameters such as the number of NFs and $\alpha^{f_{m}^{s}}$, just the value of $R_{u}$ is \textcolor{black}{variable} and depends on the value of the optimization variables. 		
Therefore the value of it has impact on the value of $\boldsymbol{\eta}=[\eta_{n}]$ that is output of  Algorithm 2. Based on (8) and server selection policy of Algorithm 2, $R_{u}$ is directly proportional to $\boldsymbol{\eta}$. Hence, if the value of $R_{u}$ is fixed or reduced at each iteration $z$, i.e., if $R_{u}^{(z)}\le R_{u}^{(z-1)}$, then we have $\boldsymbol{\eta}^{(z)}\le\boldsymbol{\eta}^{(z-1)}$.  As a result, the proposed algorithm is monotonic.

\bibliography{citation_E2E-RA}	

\begin{thebibliography}{10}
\providecommand{\url}[1]{#1}
\csname url@samestyle\endcsname
\providecommand{\newblock}{\relax}
\providecommand{\bibinfo}[2]{#2}
\providecommand{\BIBentrySTDinterwordspacing}{\spaceskip=0pt\relax}
\providecommand{\BIBentryALTinterwordstretchfactor}{4}
\providecommand{\BIBentryALTinterwordspacing}{\spaceskip=\fontdimen2\font plus
\BIBentryALTinterwordstretchfactor\fontdimen3\font minus
  \fontdimen4\font\relax}
\providecommand{\BIBforeignlanguage}[2]{{%
\expandafter\ifx\csname l@#1\endcsname\relax
\typeout{** WARNING: IEEEtran.bst: No hyphenation pattern has been}%
\typeout{** loaded for the language `#1'. Using the pattern for}%
\typeout{** the default language instead.}%
\else
\language=\csname l@#1\endcsname
\fi
#2}}
\providecommand{\BIBdecl}{\relax}
\BIBdecl

\bibitem{8320765}
I.~{Afolabi}, T.~{Taleb}, K.~{Samdanis}, A.~{Ksentini}, and H.~{Flinck},
  ``Network slicing and softwarization: A survey on principles, enabling
  technologies, and solutions,'' \emph{IEEE Communications Surveys Tutorials},
  vol.~20, no.~3, pp. 2429--2453, Mar. 2018.

\bibitem{8125672}
Z.~{Chang}, Z.~{Zhou}, S.~{Zhou}, T.~{Chen}, and T.~{Ristaniemi}, ``Towards
  service-oriented {5G}: Virtualizing the networks for
  everything-as-a-service,'' \emph{IEEE Access}, vol.~6, pp. 1480--1489, Dec.
  2018.

\bibitem{mijumbi2015design}
R.~{Mijumbi}, J.~{Serrat}, J.~{Gorricho}, N.~{Bouten}, F.~{De Turck}, and
  S.~{Davy}, ``Design and evaluation of algorithms for mapping and scheduling
  of virtual network functions,'' in \emph{Proc. IEEE Conference on Network
  Softwarization (NetSoft)}, Apr. 2015, pp. 1--9.

\bibitem{mijumbi2016network}
R.~{Mijumbi}, J.~{Serrat}, J.~{Gorricho}, N.~{Bouten}, F.~{De Turck}, and
  R.~{Boutaba}, ``Network function virtualization: {State}-of-the-art and
  research challenges,'' \emph{IEEE Communications Surveys Tutorials}, vol.~18,
  no.~1, pp. 236--262, Sep. 2016.

\bibitem{herrera2016resource}
J.~{Gil Herrera} and J.~F. {Botero}, ``Resource allocation in{ NFV: A}
  comprehensive survey,'' \emph{IEEE Transactions on Network and Service
  Management}, vol.~13, no.~3, pp. 518--532, Sep. 2016.

\bibitem{8675284}
A.~N. {Al-Quzweeni}, A.~Q. {Lawey}, T.~E.~H. {Elgorashi}, and J.~M.~H.
  {Elmirghani}, ``Optimized energy aware {5G} network function
  virtualization,'' \emph{IEEE Access}, vol.~7, pp. 44\,939--44\,958, Mar.
  2019.

\bibitem{riera2014virtual}
J.~F. {Riera}, E.~{Escalona}, J.~{Batallé}, E.~{Grasa}, and J.~A.
  {García-Espín}, ``Virtual network function scheduling: {Concept} and
  challenges,'' in \emph{Proc. 2014 International Conference on Smart
  Communications in Network Technologies (SaCoNeT)}, June 2014, pp. 1--5.

\bibitem{alliance2016description}
N.~Alliance, ``Description of network slicing concept,'' \emph{NGMN 5G P},
  vol.~1, p.~1, Jan. 2016.

\bibitem{7926921}
J.~{Ordonez-Lucena}, P.~{Ameigeiras}, D.~{Lopez}, J.~J. {Ramos-Munoz},
  J.~{Lorca}, and J.~{Folgueira}, ``Network slicing for 5{G} with {SDN}/{NFV}:
  Concepts, architectures, and challenges,'' \emph{IEEE Communications
  Magazine}, vol.~55, no.~5, pp. 80--87, May. 2017.

\bibitem{7243304}
R.~{Mijumbi}, J.~{Serrat}, J.~{Gorricho}, N.~{Bouten}, F.~{De Turck}, and
  R.~{Boutaba}, ``Network function virtualization: State-of-the-art and
  research challenges,'' \emph{IEEE Communications Surveys Tutorials}, vol.~18,
  no.~1, pp. 236--262, Sep. 2016.

\bibitem{etsi2014network1}
N.~ETSI, ``Network functions virtualisation (nfv); management and
  orchestration,'' \emph{NFV-MAN}, vol.~1, p.~v0, 2014.

\bibitem{etsi2014network}
------, ``Network functions virtualisation ({NFV}); terminology for main
  concepts in nfv,'' \emph{Group Specification, Dec}, Dec. 2018.

\bibitem{7962178}
H.~{Huang}, S.~{Guo}, J.~{Wu}, and J.~{Li}, ``Service chaining for hybrid
  network function,'' \emph{IEEE Transactions on Cloud Computing}, vol.~7,
  no.~4, Oct. 2019.

\bibitem{zeng2018stochastic}
D.~{Zeng}, J.~{Zhang}, L.~{Gu}, and S.~{Guo}, ``Stochastic scheduling towards
  cost efficient network function virtualization in edge cloud,'' in
  \emph{Proc. 2018 15th Annual IEEE International Conference on Sensing,
  Communication, and Networking (SECON)}, June 2018, pp. 1--9.

\bibitem{6782394}
M.~A. {Rodriguez} and R.~{Buyya}, ``Deadline based resource provisioningand
  scheduling algorithm for scientific workflows on clouds,'' \emph{IEEE
  Transactions on Cloud Computing}, vol.~2, no.~2, pp. 222--235, Apr. 2014.

\bibitem{7938391}
S.~{Ayoubi}, S.~{Sebbah}, and C.~{Assi}, ``A logic-based benders decomposition
  approach for the {VNF} assignment problem,'' \emph{IEEE Transactions on Cloud
  Computing}, vol.~7, no.~4, pp. 894--906, Oct. 2019.

\bibitem{8631169}
H.~{Hawilo}, M.~{Jammal}, and A.~{Shami}, ``Network function
  virtualization-aware orchestrator for service function chaining placement in
  the cloud,'' \emph{IEEE Journal on Selected Areas in Communications},
  vol.~37, no.~3, pp. 643--655, Mar. 2019.

\bibitem{7945848}
M.~{Mechtri}, C.~{Ghribi}, O.~{Soualah}, and D.~{Zeghlache}, ``{NFV}
  orchestration framework addressing {SFC} challenges,'' \emph{IEEE
  Communications Magazine}, vol.~55, no.~6, pp. 16--23, Jun. 2017.

\bibitem{cohen2015near}
R.~{Cohen}, L.~{Lewin-Eytan}, J.~S. {Naor}, and D.~{Raz}, ``Near optimal
  placement of virtual network functions,'' in \emph{Proc. IEEE Conference on
  Computer Communications (INFOCOM)}, 2015, pp. 1346--1354. Kowloon, Hong Kong.
  Apr.

\bibitem{hossain20155g}
E.~{Hossain} and M.~{Hasan}, ``{5G} cellular: key enabling technologies and
  research challenges,'' \emph{IEEE Instrumentation Measurement Magazine},
  vol.~18, no.~3, pp. 11--21, June 2015.

\bibitem{7143328}
H.~{Dahrouj}, A.~{Douik}, O.~{Dhifallah}, T.~Y. {Al-Naffouri}, and
  M.~{Alouini}, ``Resource allocation in heterogeneous cloud radio access
  networks: advances and challenges,'' \emph{IEEE Wireless Communications},
  vol.~22, no.~3, pp. 66--73, June 2015.

\bibitem{8004168}
H.~{Zhang}, N.~{Liu}, X.~{Chu}, K.~{Long}, A.~{Aghvami}, and V.~C.~M. {Leung},
  ``Network slicing based {5G} and future mobile networks: Mobility, resource
  management, and challenges,'' \emph{IEEE Communications Magazine}, vol.~55,
  no.~8, pp. 138--145, Aug. 2017.

\bibitem{7490404}
L.~{Qu}, C.~{Assi}, and K.~{Shaban}, ``Delay-aware scheduling and resource
  optimization with network function virtualization,'' \emph{IEEE Transactions
  on Communications}, vol.~64, no.~9, pp. 3746--3758, Sep. 2016.

\bibitem{game-theroy}
C.~{Pham}, N.~H. {Tran}, and C.~S. {Hong}, ``Virtual network function
  scheduling: A matching game approach,'' \emph{IEEE Communications Letters},
  vol.~22, no.~1, Jan. 2018.

\bibitem{7502870}
L.~{Qu}, C.~{Assi}, and K.~{Shaban}, ``Network function virtualization
  scheduling with transmission delay optimization,'' in \emph{in Proc IEEE/IFIP
  Network Operations and Management Symposium}, April Istanbul, Turkey, Apr.
  2016, pp. 638--644.

\bibitem{8501940}
X.~{Chen}, W.~{Ni}, I.~B. {Collings}, X.~{Wang}, and S.~{Xu}, ``Automated
  function placement and online optimization of network functions
  virtualization,'' \emph{IEEE Transactions on Communications}, vol.~67, no.~2,
  pp. 1225--1237, Feb. 2019.

\bibitem{8281644}
D.~{Li}, P.~{Hong}, K.~{Xue}, and j.~{Pei}, ``Virtual network function
  placement considering resource optimization and {SFC} requests in cloud
  datacenter,'' \emph{IEEE Transactions on Parallel and Distributed Systems},
  vol.~29, no.~7, pp. 1664--1677, July 2018.

\bibitem{7859379}
C.~{Pham}, N.~H. {Tran}, S.~{Ren}, W.~{Saad}, and C.~S. {Hong}, ``Traffic-aware
  and energy-efficient v{NF} placement for service chaining: Joint sampling and
  matching approach,'' \emph{IEEE Transactions on Services Computing}, pp.
  1--1, 2017.

\bibitem{7417401}
M.~T. {Beck} and J.~F. {Botero}, ``Coordinated allocation of service function
  chains,'' in \emph{Proc IEEE Global Communications Conference (GLOBECOM)},
  San Diego, CA, USA, Dec. 2015, pp. 1--6.

\bibitem{liu2017dynamic}
J.~{Liu}, W.~{Lu}, F.~{Zhou}, P.~{Lu}, and Z.~{Zhu}, ``On dynamic service
  function chain deployment and readjustment,'' \emph{IEEE Transactions on
  Network and Service Management}, vol.~14, no.~3, pp. 543--553, Sep. 2017.

\bibitem{8937740}
L.~{Qu}, C.~{Assi}, M.~{Khabbaz}, and Y.~{Ye}, ``Reliability-aware service
  function chaining with function decomposition and multipath routing,''
  \emph{IEEE Transactions on Network and Service Management}, pp. {},
  month={Dec. },, 2019.

\bibitem{riggio2016scheduling}
R.~{Riggio}, A.~{Bradai}, D.~{Harutyunyan}, T.~{Rasheed}, and T.~{Ahmed},
  ``Scheduling wireless virtual networks functions,'' \emph{IEEE Transactions
  on Network and Service Management}, vol.~13, no.~2, pp. 240--252, June 2016.

\bibitem{nejad2018vspace}
M.~A.~T. Nejad, S.~Parsaeefard, M.~A. Maddah-Ali, T.~Mahmoodi, and B.~H.
  Khalaj, ``{vSPACE}: {VNF} simultaneous placement, admission control and
  embedding,'' \emph{IEEE Journal on Selected Areas in Communications},
  vol.~36, no.~3, pp. 542--557, Mar. 2018.

\bibitem{kim2018performance}
H.~Kim, ``Performance evaluation of revised virtual resources allocation scheme
  in network function virtualization ({NFV}) networks,'' \emph{Cluster
  Computing}, vol.~22, no.~1, pp. 2331--2339, 2019.

\bibitem{8460139}
T.~{Ahmed}, A.~{Alleg}, R.~{Ferrus}, and R.~{Riggio}, ``On-demand network
  slicing using {SDN/NFV}-enabled satellite ground segment systems,'' in
  \emph{Proc IEEE Conference on Network Softwarization and Workshops
  (NetSoft)}, Montreal, QC, Canada, Jun. 2018, pp. 242--246.

\bibitem{8845306}
M.~{Femminella} and G.~{Reali}, ``Gossip-based monitoring of virtualized
  resources in 5g networks,'' in \emph{Proc IEEE Conference on Computer
  Communications Workshops (INFOCOM WKSHPS)}, Paris, France, France, Apr. 2019,
  pp. 378--384.

\bibitem{8647504}
R.~A. {Addad}, T.~{Taleb}, M.~{Bagaa}, D.~L.~C. {Dutra}, and H.~{Flinck},
  ``Towards modeling cross-domain network slices for {5G},'' in \emph{in Proc
  IEEE Global Communications Conference (GLOBECOM)}, Abu Dhabi, United Arab
  Emirates, Dec. 2018, pp. 1--7.

\bibitem{7835175}
Y.~{Mansouri}, A.~N. {Toosi}, and R.~{Buyya}, ``Cost optimization for dynamic
  replication and migration of data in cloud data centers,'' \emph{IEEE
  Transactions on Cloud Computing}, vol.~7, no.~3, pp. 705--718, Jul. 2019.

\bibitem{7934437}
H.~{Shah-Mansouri}, V.~W.~S. {Wong}, and R.~{Schober}, ``Joint optimal pricing
  and task scheduling in mobile cloud computing systems,'' \emph{IEEE
  Transactions on Wireless Communications}, vol.~16, no.~8, pp. 5218--5232,
  Aug. 2017.

\bibitem{8247219}
B.~{Yang}, W.~K. {Chai}, Z.~{Xu}, K.~V. {Katsaros}, and G.~{Pavlou},
  ``Cost-efficient {NFV}-enabled mobile edge-cloud for low latency mobile
  applications,'' \emph{IEEE Transactions on Network and Service Management},
  vol.~15, no.~1, pp. 475--488, Mar. 2018.

\bibitem{8638582}
Y.~{Chen}, N.~{Zhang}, Y.~{Zhang}, X.~{Chen}, W.~{Wu}, and X.~S. {Shen},
  ``Energy efficient dynamic offloading in mobile edge computing for internet
  of things,'' \emph{IEEE Transactions on Cloud Computing}, 2019.

\bibitem{8672634}
B.~{Ren}, D.~{Guo}, Y.~{Shen}, G.~{Tang}, and X.~{Lin}, ``Embedding service
  function tree with minimum cost for {NFV}-enabled multicast,'' \emph{IEEE
  Journal on Selected Areas in Communications}, vol.~37, no.~5, pp. 1085--1097,
  May. 2019.

\bibitem{7949048}
N.~{Yu}, Z.~{Song}, H.~{Du}, H.~{Huang}, and X.~{Jia}, ``Dynamic resource
  provisioning for energy efficient cloud radio access networks,'' \emph{IEEE
  Transactions on Cloud Computing}, vol.~7, no.~4, pp. 964--974, Oct. 2019.

\bibitem{9000731}
J.~{Chen}, H.~{Liu}, and H.~{Jia}, ``Cross-layer resource allocation in
  wireless-enabled {NFV},'' \emph{IEEE Wireless Communications Letters}, pp.
  1--1, IEEE Early Access, Feb. 2020.

\bibitem{7833146}
P.~D. {Diamantoulakis}, K.~N. {Pappi}, G.~K. {Karagiannidis}, H.~{Xing}, and
  A.~{Nallanathan}, ``Joint downlink/uplink design for wireless powered
  networks with interference,'' \emph{IEEE Access}, vol.~5, pp. 1534--1547,
  Jan. 2017.

\bibitem{ETSIG003}
{ETSI, GSNFV}, ``Network functions virtualisation ({NFV}); terminology for main
  concepts in {NFV},'' \emph{ETSI GS NFV}, vol.~2, no.~2, p.~V1, Aug. 2018.

\bibitem{etsi2013network}
------, ``Network functions virtualisation ({NFV}): Architectural framework,''
  \emph{ETSI GS NFV}, vol.~2, no.~2, p.~V1, 2018.

\bibitem{yoon2016nfv}
M.~S. Yoon and A.~E. Kamal, ``{NFV} resource allocation using mixed queuing
  network model,'' in \emph{in Proc. IEEE Global Communications Conference
  (GLOBECOM)}.\hskip 1em plus 0.5em minus 0.4em\relax IEEE, 2016, pp.
  Washington, DC, USA, 1--6. Dec. 2016.

\bibitem{28530}
{3GPP, TS 28.530}, ``Management and orchestration; concepts, use cases and
  requirements, {R}elease 15,'' Mar. 2019.

\bibitem{fajardo2015improving}
J.~O. Fajardo, I.~Taboada, and F.~Liberal, ``Improving content delivery
  efficiency through multi-layer mobile edge adaptation,'' \emph{IEEE Network},
  vol.~29, no.~6, pp. 40--46, 2015.

\bibitem{qu2016delay}
L.~Qu, C.~Assi, and K.~Shaban, ``Delay-aware scheduling and resource
  optimization with network function virtualization,'' \emph{IEEE Transactions
  on Communications}, vol.~64, no.~9, pp. 3746--3758, 2016.

\bibitem{chinneck2007feasibility}
J.~W. Chinneck, \emph{Feasibility and Infeasibility in Optimization::
  Algorithms and Computational Methods}.\hskip 1em plus 0.5em minus 0.4em\relax
  Springer Science \& Business Media, 2007, vol. 118.

\bibitem{ebrahimi2019joint}
S.~Ebrahimi, A.~Zakeri, B.~Akbari, and N.~Mokari, ``Joint resource and
  admission management for slice-enabled networks,'' \emph{arXiv}, pp.
  arXiv--1912, Accepted to Puplish in IEEE/IFIP NOSM, 2020.

\bibitem{grant2006matlab}
M.~Grant, S.~Boyd, and Y.~Y. {CVX}, ``Matlab software for disciplined convex
  programming, version 1.0 beta 3,'' \emph{Recent Advances in Learning and
  Control}, pp. 95--110, 2006.

\bibitem{mosek2015mosek}
A.~Mosek, ``The {MOSEK} optimization toolbox for {MATLAB} manual,'' 2015.

\bibitem{8786250}
A.~{Zakeri}, M.~{Moltafet}, and N.~{Mokari}, ``Joint radio resource allocation
  and sic ordering in {NOMA}-based networks using submodularity and matching
  theory,'' \emph{IEEE Transactions on Vehicular Technology}, vol.~68, no.~10,
  pp. 9761--9773, Oct. 2019.

\bibitem{7899529}
T.~M. {Ho}, N.~H. {Tran}, S.~M. {Ahsan Kazmi}, and C.~S. {Hong}, ``Dynamic
  pricing for resource allocation in wireless network virtualization: A
  stackelberg game approach,'' in \emph{2017 International Conference on
  Information Networking (ICOIN)}, Da Nang, Vietnam, Apr. 2017, pp. 429--434.

\bibitem{osseiran2015manufacturing}
S.~{Misra}, A.~{Mondal}, and S.~{Khajjayam}, ``Dynamic big-data broadcast in
  fat-tree data center networks with mobile {IoT} devices,'' \emph{IEEE Systems
  Journal}, pp. 1--8, Mar. 2019.

\bibitem{6968961}
S.~{Mehraghdam}, M.~{Keller}, and H.~{Karl}, ``Specifying and placing chains of
  virtual network functions,'' in \emph{Proc. IEEE 3rd International Conference
  on Cloud Networking (CloudNet)}, Oct. 2014, pp. 7--13.

\bibitem{4786509}
J.~{Huang}, V.~G. {Subramanian}, R.~{Agrawal}, and R.~A. {Berry}, ``Downlink
  scheduling and resource allocation for {OFDM} systems,'' \emph{IEEE
  Transactions on Wireless Communications}, vol.~8, no.~1, pp. 288--296, Jan,
  2009.

\bibitem{8654611}
Q.~{Wang} and F.~{Zhao}, ``Joint spectrum and power allocation for {NOMA}
  enhanced relaying networks,'' \emph{IEEE Access}, vol.~7, pp.
  27\,008--27\,016, Feb. 2019.

\end{thebibliography}
\bibliographystyle{ieeetran}
\end{document}